\documentclass[11pt,a4paper]{article}

\usepackage[a4paper, total={5.8in, 9in}]{geometry}

\usepackage{amsmath,amsfonts}
\usepackage{amsthm}
\usepackage{amssymb}
\usepackage[english]{babel}
\usepackage{enumitem}
\usepackage{mdframed}
\usepackage{xspace}
\usepackage{todonotes}

\usepackage{soul}

\usepackage[ruled,linesnumbered,vlined]{algorithm2e}
\usepackage{epsfig}
\usepackage{bbm}
\usepackage{thm-restate}
\usepackage[colorlinks,linkcolor=Darkblue,filecolor=blue,citecolor=blue,urlcolor=Darkblue,pagebackref]{hyperref}
\usepackage[nameinlink]{cleveref}

\usepackage{xcolor} 
\hypersetup{
    colorlinks,
    linkcolor={blue},
    citecolor={blue},
    urlcolor={blue}
}

\newtheorem{theorem}{Theorem}[section]
\newtheorem{lemma}[theorem]{Lemma}

\newtheorem{claim}{Claim}

\numberwithin{theorem}{section}

\newtheorem{remark}[theorem]{Remark}

\newcommand{\cupdot}{\mathbin{\mathaccent\cdot\cup}}

\newcommand{\pedge}{{$+$edge}\xspace}

\newcommand{\cor}{\textsc{Correlation Clustering}\xspace}

\newcommand{\ppmnotbad}{1.946}
\newcommand{\ppmbad}{2}
\newcommand{\mmmrat}{1}
\newcommand{\ppprat}{1.5}
\newcommand{\pmmrat}{1.5}

\newcommand{\ppprats}{1.9}
\newcommand{\pmmrats}{1.65}

\newcommand{\degrat}{1}

\newenvironment{cproof}
{\begin{proof}
 [Proof.]
 \vspace{-1.5\parsep}%-3.2\parsep} %%%For use with US letter
}
{ \end{proof}}

\newcommand{\cc}{\textsc{Correlation Clustering}\xspace}

\def\eps{\varepsilon}

\def\calT{{\cal T}}
\def\calD{{\cal D}}
\def\calR{{\cal R}}
\def\calB{{\cal B}}
\def\calC{{\cal C}}
\def\1{\mathbb{I}}

\newcommand{\E}{{\mathbb{E}}}
\newcommand{\R}{{\mathbb{R}}}

\newcommand{\dx}{\frac{x^2}{\delta}}
\newcommand{\dy}{\frac{y^2}{\delta}}
\newcommand{\dz}{\frac{z^2}{\delta}}

\newcommand{\threshold}{0.1}
\newcommand{\ratio}{1.994}
\newcommand{\finaleta}{1/12}
\newcommand{\finalgamma}{0.054}

\newcommand{\cost}{\mathsf{cost}}
\newcommand{\costr}{\mathsf{cost}^r}
\newcommand{\costi}{\mathsf{cost}^i}
\newcommand{\costs}{\mathsf{cost}^s}
\newcommand{\costf}{\mathsf{cost}^f}

\newcommand{\lp}{\mathsf{lp}}
\newcommand{\lpr}{\mathsf{lp}^r}
\newcommand{\lpi}{\mathsf{lp}^i}
\newcommand{\lps}{\mathsf{lp}^s}
\newcommand{\lpf}{\mathsf{lp}^f}

\begin{document} 
\title{Correlation Clustering with Sherali-Adams}

 \author{
Vincent Cohen-Addad\thanks{Google Research.}
\and Euiwoong Lee\thanks{University of Michigan.}
\and Alantha Newman\thanks{Laboratoire G-SCOP (CNRS, Grenoble-INP).
    Supported in part by French ANR Project DAGDigDec (ANR-21-CE48-0012).}
}

\date{}

\maketitle

\begin{abstract}

Given a complete graph $G = (V, E)$ where each edge is labeled $+$ or
$-$, the \cc problem asks to partition $V$ into clusters to minimize
the number of $+$edges between different clusters plus the number of
$-$edges within the same cluster.  \cc has been used to model a large
number of clustering problems in practice, making it one of the most
widely studied clustering formulations. The approximability of \cc has
been actively investigated~\cite{BBC04, CGW05, ACN08}, culminating in
a $2.06$-approximation algorithm~\cite{CMSY15}, based on rounding the
standard LP relaxation.  Since the integrality gap for this
formulation is 2, it has remained a major open question to determine
if the approximation factor of 2 can be reached, or even breached.

In this paper, we answer this question affirmatively by showing that
there exists a $(\ratio + \eps)$-approximation algorithm based on
$O(1/\eps^2$) rounds of the Sherali-Adams hierarchy. In order to round
a solution to the Sherali-Adams relaxation, we adapt the {\em
  correlated rounding} originally developed for CSPs~\cite{BRS11,
  GS11, RT12}.  With this tool, we reach an approximation ratio of
$2+\eps$ for \cor.  To breach this ratio, we go beyond the traditional
triangle-based analysis by employing a {\em global charging scheme}
that amortizes the total cost of the rounding across different
triangles.
\end{abstract}

\section{Introduction}
Clustering is a central problem in unsupervised machine learning and
data mining. Given a dataset and information regarding the similarity
of pairs of elements, a ``good'' clustering is a partition of the
elements into groups such that similar elements belong to the same
group, while dissimilar elements belong to different groups.  Since
its introduction by Bansal, Blum, and Chawla~\cite{BBC04}, \cor has
been one of the most widely studied formulations for clustering. Given
a graph $G = (V, E)$ where each edge is either labeled $+$ or $-$, the
goal is to find a clustering (partition) $(V_1, \dots, V_k)$ of $V$
that minimizes the number of unsatisfied edges, namely the $+$edges
between different clusters and the $-$edges within the same cluster.
Thanks to the simplicity and modularity of the formulation, \cor has
found a spectacular number of applications, e.g., finding clustering
ensembles \cite{bonchi2013overlapping}, duplicate detection
\cite{arasu2009large}, community mining \cite{chen2012clustering},
disambiguation tasks \cite{kalashnikov2008web}, automated labelling
\cite{agrawal2009generating, chakrabarti2008graph} and many more.

When $G$ is a general graph, there is an $O(\log n$)-approximation
algorithm~\cite{CGW05, DEFI06} via the equivalence to
\textsc{Undirected Multicut}, which is hard to approximate within any
constant factor assuming the Unique Games Conjecture
(UGC)~\cite{CKKRS06}. For the maximization version where the goal is
to maximize the number of $+$edges within the same cluster plus the
number of $-$edges between different clusters, Charikar, Guruswami,
and Wirth~\cite{CGW05} and Swamy~\cite{Swamy04} gave
$0.766$-approximation algorithms based on rounding semidefinite
programs.

A lot of effort has focused on understanding the approximability of
the original version introduced by \cite{BBC04}: the unweighted case
on a complete graph.  (For the rest of the paper, \cor denotes this
version.)  In this case, \cite{BBC04} gave a PTAS for the maximization
version and an $O(1)$-approximation for the minimization version.
Charikar, Guruswami, and Wirth gave a 4-approximation based on
rounding the standard linear programming (LP) relaxation and proved
APX-hardness~\cite{CGW05}.  Ailon, Charikar, and Newman gave a
combinatorial 3-approximation algorithm based on choosing random
pivots and a 2.5-approximation by combining this pivot based approach
with the standard LP relaxation~\cite{ACN08}.
  
The current best approximation ratio in this classic setting is
$2.06-\eps$ for some fixed $\eps > 0$ by Chawla, Makarychev, Schramm,
and Yaroslavtsev~\cite{CMSY15}, which extended the pivot rounding
framework of \cite{ACN08} with advanced functions that convert LP
values to rounding probabilities.  The standard LP relaxation for \cor
has an integrality gap of $2$~\cite{CGW05}.  The LP-based
approximation algorithms (i.e., the 4-approximation of \cite{CGW05},
the 2.5-approximation algorithm of~\cite{ACN08}, and the
2.06-approximation algorithm of~\cite{CMSY15}) each prove upper bounds
on the integrality gap of this LP.  Furthermore,~\cite{CMSY15} shows
that their rounding framework cannot yield an approximation ratio
better than $2.025$.

Thus, currently, even reaching the approximation threshold of 2 is an
interesting open problem.  In this paper, we overcome the
aforementioned barriers and give a $(\ratio + \eps)$-approximation
algorithm for \cor for any $\eps > 0$ using the Sherali-Adams
hierarchy.

\begin{theorem}
  \label{thm:main}
  For $\eps > 0$, there exists a $(\ratio + \eps)$-approximation algorithm for \cc  running in time $n^{O(1/\eps^2)}$.
  Moreoever, the integrality gap of the $O(1/\eps^2)$-round Sherali-Adams relaxation is at most $(\ratio + \eps)$.
\end{theorem}

While we will present our algorithm as a randomized algorithm, it can
be derandomized using the standard method of conditional
expectation. See Section~\ref{sec:derandomization} for details. In
order to achieve the result, we introduce the following two techniques
for \cor.  Our result also implies a marginally better constant factor
approximation for the problem of fitting a tree metric or an
ultrametric through the framework
of~\cite{DBLP:conf/focs/Cohen-Addad0KPT21,DBLP:journals/siamcomp/AilonC11}.

\begin{itemize}
\item To improve beyond the integrality gap of the standard LP, we
  naturally use the {\em Sherali-Adams hierarchy} tailored for \cor,
  defined in Section~\ref{sec:sa}. Previous algorithms~\cite{ACN08,
    CMSY15} proceed by sampling a random pivot $p \in V$ in each
  iteration and independently deciding whether $v \in V \setminus \{ p
  \}$ belongs to $p$'s cluster or not. In order to use the power of
  Sherali-Adams, we adapt the {\em correlated rounding} that has been
  used for \textsc{Constraint Satisfaction Problems}
  (CSPs)~\cite{BRS11, GS11, RT12}. One of the main advantages is that
  for $u, v \in V \setminus \{ p \}$, one can ensure that $\Pr[u\mbox{
      and }v\mbox{ belong to }p\mbox{'s cluster}]$ is approximately
  equal to the value predicted by the Sherali-Adams solution in an
  amortized sense. See Section~\ref{sec:algo} for the description of
  the algorithm.

\item Previous analyses~\cite{ACN08, CMSY15} employ the elegant {\em triangle-based} analysis that bounds the ratio 
\begin{equation}
\frac{\cost_u(v, w) + \cost_v(w, u) + \cost_w(u, v)}{\lp_u(v, w) + \lp_v(w, u) + \lp_w(u, v)},
\label{eq:basicratio}
\end{equation}
for each triangle $(u, v, w)$, where $\cost_u(v, w)$ is the
probability that $(v, w)$ is violated when $u$ is pivot, and $\lp_u(v,
w)$ is the probability that $v$ or $w$ belongs to $u$'s cluster
(removing $(v, w)$ from the instance) times the LP contribution of
$(v, w)$. (See Section~\ref{subsec:setup} for this basic setup.) The
previous analyses bound the above ratio for every triangle
individually. With our rounding algorithm, the
ratio~\eqref{eq:basicratio} is already at most $2$ for every triangle,
and there is only one type of a {\em bad triangle} that has a ratio
close to $2$ (i.e., $++-$ triangles with LP values close to $0.5, 0.5,
1$ respectively, which are bad triangles for the previous rounding
algorithms as well). We prove that the number of such bad triangles is
not large compared to {\em chargeable triangles} that have
significantly smaller ratios but still with large denominators. This
allows a global charging scheme where we show that the total ratio
(the sum of numerators over all triangles / the sum of denominators
over all triangles) is strictly less than
$2$. Sections~\ref{sec:charging} and~\ref{subsec:finish} show how we
use this scheme to finish the analysis.

\end{itemize}

\subsection{Further Related Work}

The pivot-based algorithm of Ailon et al.~\cite{ACN08} has been
revisited in terms of derandomization~\cite{van2009deterministic},
parallelism~\cite{chierichetti2014correlation}, for classification
with asymmetric error costs~\cite{jafarov2020correlation}, and for
clustering with categorical rather than binary relationships given
between elements~\cite{anava2015improved,bonchi2013overlapping}, to
name a few settings in which it has been applied and adapted.
A related objective function which maximizes the difference between
the satisfied and unsatisfied edges has been
studied~\cite{charikar2004maximizing,alon2006quadratic}.
The
\cor problem has also been studied in an online
setting~\cite{mathieu2010online} and with respect to local
guarantees~\cite{puleo2016correlation,charikar2017local,kalhan2019correlation,jafarov2021local}.
Recent progress has lead to constant factor approximation
algorithms for the problem in the massively-parallel computation model~\cite{cohen2021correlation,abs-2205-03710}, in the streaming setting~\cite{DBLP:conf/innovations/Assadi022}, online
setting~\cite{Cohen-AddadLMP22}, and with differential privacy
guarantees~\cite{bun2021differentially,DBLP:journals/corr/abs-2203-01440,
  Daogao2022}.

Besides complete graphs, other special classes of graphs have been
considered, including complete $k$-partite
graphs~\cite{ailon2012improved, CMSY15} and the weighted case where
the weights of $-$edges satisfy the triangle
inequality~\cite{gionis2007clustering}.  The result for complete
$k$-bipartite graphs match the integrality gap of the standard
LP. When the number of clusters $k$ is bounded, Giotis and
Guruswami~\cite{giotis2006correlation} and Karpinski and
Schudy~\cite{karpinski2009linear} showed that a PTAS exists.

In terms of using the Sherali-Adams hierarchy to design approximation
algorithms, there have been numerous negative
results~\cite{charikar2009integrality,georgiou2009optimal,karlin2011integrality}
as well as some applications for designing
algorithms~\cite{yoshida2014approximation,aprile2020simple,o2019sherali,hopkins2020subexponential}.

\section{Preliminaries}

An instance of \cc is a complete graph $G = (V, E)$, where $E = E^+
\cup E^-$ and $E^+\cap E^- = \emptyset$, and the goal is to compute a
partition $\{C_1,\ldots,C_k\}$ of $V$ minimizing the number of the
$+$edges $(u,v)$ where $u \in C_i$ and $v \in C_j$, $i \neq j$, plus
the number of the $-$edges $(u,v)$ where $u,v \in C_i$.  The following
standard LP relaxation has been used by most of the previous
work~\cite{CGW05, ACN08, CMSY15}.

\begin{align*}
  \min \sum_{ij \in E^+}x_{ij} &+ \sum_{ij \in E^-} (1-x_{ij})\\
   x_{ij} & \leq x_{ik} + x_{jk} \quad \forall i,j,k \in V\\
  x_{ij} & \geq 0 \quad \forall i,j \in V.
\end{align*}

It has an integrality gap of $2$~\cite{CGW05}; consider a graph with
vertices $\{ 0, 1, \dots, k \}$ where the edge $(0, i)$ is $+$ for
each $i \in [k]$ and the rest are $-$. Letting $x_{0i} = 1/2$ for each
$i \in [k]$ and $x_{ij} = 1$ for each $1 \leq i < j \leq k$ ensures
that the LP value is $k/2$, but the optimal integral value is $k - 1$.

\subsection{Strengthened LP relaxation}
\label{sec:sa}

In order to overcome the integrality gap for the standard LP
relaxation, we consider the following $r$-rounds of Sherali-Adams
relaxation.  For a collection of nonempty disjoint sets $S_1, \dots,
S_{\ell} \subseteq V$ such that $\sum_{i=1}^{\ell} |S_i| \leq r$, we
have a variable $y_{S_1 | S_2 | \dots | S_{\ell}}$ indicating the
probability that the optimal partition induced by $S_1 \cup \dots \cup
S_{\ell}$ is exactly $(S_1, \dots, S_{\ell})$.  Note that the order of
$S_1, \dots, S_{\ell}$ does not matter. For example, for two vertices
$u$ and $v$, $y_{u|v}$ is supposed to indicate the probability that
$u$ and $v$ are in different clusters in the optimal solution and
$y_{uv}$ indicates the probability that they are in the same cluster
so that $y_{u|v} + y_{uv} = 1$. (Similarly, for three distinct
vertices $u, v, w$, $y_{uvw} + y_{u|vw} + y_{v|uw} + y_{w|uv} +
y_{u|v|w} = 1$.) We have the following constraints ensuring the
consistency of the variables.  We use $S_i \cupdot S_j$ to indicate
disjoint unions. Notice that $x_{uv} = 1- y_{uv} = y_{u|v}$.

\begin{align}
  \min \quad & \sum_{ij \in E^+} (1-y_{ij})  + \sum_{ij \in E^-} y_{ij} \\
\mbox{s.t.} \quad & y_{T_1|\ldots|T_k}   = \sum_{\substack{S_1,\ldots, S_{\ell}: \\ S = S_1 \cupdot \ldots \cupdot S_{\ell} \\ \text{and } T_i = S_i \cap T~ \forall i \in [k]}} y_{S_1|S_2|\ldots|S_{\ell}} \quad && \forall 
T \subseteq  S \subseteq V, |S| \leq r, \text{ and } T= T_1 \cupdot \ldots \cupdot T_k\\
  &y_{\emptyset} = 1  \\
  &y  \geq 0. \label{nonneg_y}
\end{align}
Note that the constraint \eqref{nonneg_y} requires the $y$-variables across all possible subscripts to be nonnegative.

\section{Algorithm}
\label{sec:algo}
Let $r$ be a positive integer denoting the number of rounds, and $\delta := \threshold$ throughout the paper. 
We consider the solution $y$ obtained from $r$-round of Sherali-Adams. 
Let $x_{uv} := y_{u|v} = 1 - y_{uv}$ be the {\em distance} between $u$ and $v$. 
Call an edge $(u, v)$ {\em short} if $x_{uv} \leq \delta$, {\em long} if $x_{uv} \geq 1 - \delta$, and {\em medium} otherwise.

\subsection{Rounding Algorithm}
At a high-level, our rounding algorithm follows the general framework of~\cite{ACN08} and~\cite{CMSY15}. 
The algorithm proceeds in iterations, and in each iteration with the remaining instance $G = (V, E)$, the algorithm chooses a pivot $p$ uniformly at random from $V$, samples a random set $S \ni p$, creates $S$ as a new cluster, and proceeds with the remaining instance $G \setminus S$. 

The most crucial step of the algorithm is to sample $S$ in each
iteration.  Given a pivot $p \in V,$ for each vertex $v \in V
\setminus {p},$ the \textsc{LP-KwikCluster} algorithm~\cite{ACN08}
independently puts $v$ into $S$ with probability $(1 - x_{pv})$.  The
refined algorithm of~\cite{CMSY15} also does independent rounding,
but puts $v$ into $S$ with probability $(1 - f^{s}(x_{pv}))$ where $s
\in \{ +, - \}$ is the sign of $(p, v)$. (It sets $f^{-}(x) = x$ and
$f^{+}(x) = 0$ if $x < 0.19$, $(\frac{x-0.19}{0.5095-0.19})^2$ if $x \in [0.19,
  0.5095]$, and $1$ if $x \geq 0.5095$.)  In order to use the power of
Sherali-Adams, given a pivot $p$, we round medium $+$edges in a
correlated manner while also employing nontrivial $f^s(\cdot)$
functions for other edges. The full algorithm is described as
Algorithm~\ref{algo}.

Before we present the full algorithm, we briefly discuss some intuition behind our rounding.
  In \cite{ACN08} and \cite{CMSY15}, the analysis boiled
  down to analyzing the ratio defined in~\eqref{eq:basicratio} on each type of triangle.
  We call a triangle $+++$ if it has three $+$edges, and $++-$, $+--$,
  $---$ triangles are defined similarly.  For the LP rounding
  algorithm of \cite{ACN08}, each triangle has a ratio of at most 2,
  except the $++-$ triangle, which has a ratio of $2.5$.  In
  \cite{CMSY15}, there is a trade-off by lowering the ratio for the
  $++-$ triangles but increasing the ratio for the $+++$ triangles, to
  the point where each ratio is around $2.06$.

  A key observation is
  that we can use the {\em correlated rounding} (Line~\ref{line:correlated} of the algorithm) to lower the ratio in the case where the $++-$ triangle had ratio $2.5$
  (without increasing the ratio on other types of triangles such as $+++$).  This is because in the correlated rounding,
  we obtain a sort of {\em negative correlation}, which is not present
  in the independent rounding.  Specifically, for a $++-$ triangle
  with distances $(.5,.5,1)$ (where the $-$edge has distance 1 and
  vertex $p$ is incident to the two $+$edges), if we do the
  independent rounding with $p$ as a pivot, then there is a $1/4$
  probability that both of the other two vertices, call them $u$ and
  $v$, will be included in $p$'s cluster and therefore a $1/4$ probability of the bad event that the $-$edge will be an intracluster edge.  
  However, 
  with correlated rounding,
  if $p$ is the pivot, then the events of including $u$ and $v$ in
  $p$'s cluster are negatively correlated and exactly one of them is
  included. 
  Thus, the $-$edge, which contributes 0 to the objective
  function of the LP, is never contained in a cluster when $p$ is the
  pivot, which results in a lower ratio for this triangle.
  (In reality, it happens in an approximate and amortized sense, which slightly complicates the analysis.)

  One more comment is that we cannot simply use correlated rounding on both
  $+$ and $-$edges incident to the chosen pivot, because this turns
  out to have an unbounded ratio on $---$ triangles.  
  There are additional technical reasons that prevent us from using correlated rounding to short or long $+$edges (see Remark~\ref{remark:short}), 
  so our algorithm only uses correlated rounding on medium $+$edges incident to the pivot.

\begin{algorithm}[h]
  \caption{Rounding procedure \texttt{Round} for Correlation Clustering with parameter $\delta$.    \label{algo}}
    \label{algo:rounding}
    \textbf{Input:} Set of vertices $V$, with edges $E^+$ and $E^-$, a fractional solution $y$ from the $r$-round Sherali-Adams relaxation\;
    Pick a pivot $p \in V$ uniformly at random\;
    $S \leftarrow \{p\}$\;
    \ForEach{vertex $v \in V-\{p\}$}{
      \If{$(v,p) \in E^-$}{
        Add $v$ to $S$ independently with probability  $1 - \sqrt{x_{pv}}$\;
      }
      \If{$(v,p) \in E^+$ and $(v,p)$ is short}{\label{line:shortone}
        Add $v$ to $S$ independently with probability $1 - x_{pv}^2 / \delta$\; \label{line:shorttwo}
      }
      \If{$(v,p) \in E^+$ and $(v,p)$ is long}{\label{line:longone}
        Add $v$ to $S$ independently with probability $1 - x_{pv}$\; \label{line:longtwo}
      }
    }
    Define $I_p \gets \{ v \in V \setminus \{p\} : (v, p)\mbox{ is medium}+\mbox{edge} \}$\; \label{line:Ip}
    Sample $S' \subseteq I_p$ as prescribed by \Cref{assump:rounding}
    (i.e.: such that: (1) For each $v \in I_p$, $\Pr[v \in S'] = y_{pv}$; and (2)
    $\E_{u, v \in I_p}[|\Pr[u, v \in S'] - y_{puv}|] \leq \eps_r$, where $\eps_r = O(1/\sqrt{r})$)\; \label{line:correlated}
    $S \gets S \cup S'$\;
    \textbf{Output:} Cluster $S$ and the clusters obtained by calling \texttt{Round} on $V - S$ (with the
    fractional solution $y$ and $+$ and $-$edges induced by $V- S$)\;
\end{algorithm}

The following lemma shows that the correlated rounding procedure in Line~\ref{line:correlated} can be implemented using the techniques to round convex hierarchies for CSPs~\cite{RT12, GS11, BRS11}. It is proved in Section~\ref{sec:rt}.

\begin{lemma}
  In Line~\ref{line:correlated}, one can sample $S' \subseteq I_p$ in time $n^{O(r)}$ such that 
  \begin{itemize}
  \item For each $v \in I_p$, $\Pr[v \in S'] = y_{pv}$. 
  \item $\E_{u, v \in I_p}[|\Pr[u, v \in S'] - y_{puv}|] \leq \eps_r$, where $\eps_r = O(1/\sqrt{r})$.
  \end{itemize}
  \label{assump:rounding}
\end{lemma}

\section{Analysis}
In this section, we show that Algorithm~\ref{algo:rounding} guarantees a $(\ratio + \eps_r)$-approximation.

\label{sec:analysis}
\subsection{Setup and Ideal Cases}
\label{subsec:setup}
Our high-level setup of the analysis also follows from that
of~\cite{ACN08} and~\cite{CMSY15}.  Consider the $t$-th iteration of
Algorithm~\ref{algo} with the current graph $G_t = (V_t, E_t)$.  Let
$\costr_p(u, v)$ be the probability that $(u, v)$ is violated in the
rounding algorithm when $p$ is the pivot, and $\lpr_p(u, v)$ be the LP
value of $(u, v)$ (i.e., $x_{uv}$ if $(u, v)$ is $+$ and $y_{uv}$ if
it is $-$) times the probability that $(u, v)$ disappears (i.e.,
$\Pr[S \cap \{ u, v \} \neq \emptyset]$).  The superscript $r$ stands for rounding.

We call a set of three
distinct vertices a {\em triangle}. A set of two vertices is called a
{\em degenerate triangle}.  For triangle $\{ u, v, w \}$, let
$\costr(u, v, w) = \costr_u(v, w) + \costr_v(u, w) + \costr_w(u, v)$
and $\lpr(u, v, w) = \lpr_u(v, w) + \lpr_v(u, w) + \lpr_w(u, v)$. For
degenerate triangle $\{ u, v \}$, let $\costr(u, v) = \costr_u(u, v) +
\costr_v(u, v)$ and $\lpr(u, v) = \lpr_u(u, v) + \lpr_v(u, v)$. Let
\[
ALG_t := \E_{u \in V} \sum_{(v, w) \in \binom{V_t}{2}} \costr_u(v, w)
\]
be the expected cost incurred by this iteration, and 
\[
LP_t:= \E_{u \in V} \sum_{(v, w) \in \binom{V_t}{2}} \lpr_u(v, w)
\]
be the expected amount of the LP value removed by this iteration.
If we could show that for all $t$, 
\begin{equation}
ALG_t \leq \alpha \cdot LP_t
\label{eq:ratio}
\end{equation}
then we will get an upper bound on the total cost $\mathbf{ALG}$ as
\[
\E[\mathbf{ALG}] = \E[\sum_{t=0}^R ALG_t] \leq \alpha \cdot \E[\sum_{t=0}^R LP_t] = \alpha \cdot \mathbf{LP}
\]
where $\mathbf{LP}$ denotes the total $LP$ value and $R$ is the number of the iterations.

Therefore, in order to prove Theorem~\ref{thm:main}, it suffices to consider one iteration. For the rest of the paper, let us omit the subscript $t$ denoting the iteration. We prove~\eqref{eq:ratio}, which is equivalent to upper bounding  
\begin{align*}
\frac{ALG}{LP} = 
\frac{\E_{u \in V} \sum_{(v, w) \in E} \costr_u(v, w)}{\E_{u \in V} \sum_{(v, w) \in E} \lpr_u(v, w)} = 
\frac{\sum_{(u, v, w) \in \binom{V}{3}} \costr(u,v, w) + \sum_{(u, v) \in \binom{V}{2}} \costr(u,v)}{\sum_{(u, v, w) \in \binom{V}{3}} \lpr(u,v, w) + \sum_{(u, v) \in \binom{V}{2}} \lpr(u,v)}.
\end{align*}

Recall that a triangle is $+++$ if it has three $+$edges and $++-$,
$+--$, $---$ triangles are defined similarly.  For a degenerate
triangle $\{ u, v \}$, $\costr_u(u, v)$ and $\lpr_u(u, v)$ depend only
on $x_{uv}$ and the sign of $(u, v)$.  Even for a triangle $\{ u, v, w
\}$, the values of $\costr_u(v, w)$ and $\lpr_u(v, w)$ only depend on
$x_{uv}, x_{uw}, x_{vw}$ and the signs of the edges unless both $(u,
v)$ and $(u, w)$ are medium $+$edges; $v$ and $w$ are added to $S \cup
S'$ independently with the probabilities depending on $x_{uv}$ and
$x_{uw}$ respectively. When both $(u, v)$ and $(u, w)$ are medium
$+$edges, then they are rounded with correlation and $\Pr[v, w \in S'
  | u\mbox{ is pivot}]$ must be, ideally, exactly equal to $y_{uvw}$,
but Lemma~\ref{assump:rounding} only gives an approximate guarantee
amortized over the vertices in $I_u$.

To gradually overcome the complication arising from correlated rounding, we define the following two idealized versions of $\costr(\cdot)$ and $\lpr(\cdot)$ and analyze them first. 
\begin{itemize}
\item $\costs(\cdot)$ and $\lps(\cdot)$ are defined assuming that the correlated rounding for medium $+$edges are perfect. Formally, $\costs_u(\cdot), \costs(\cdot), \lps_u(\cdot), \lps(\cdot)$ are defined identically to $\costr_u(\cdot)$, $\costr(\cdot)$, $\lpr_u(\cdot)$, $\lpr(\cdot)$ respectively, assuming that in Line~\ref{line:correlated} of Algorithm~\ref{algo}, the condition (2) is replaced by $\Pr[u, v \in S' | p\mbox{ is pivot}] = y_{puv}$ for every $p \in V$, $u, v \in I_p$. With this assumption, note that for every  triangle $\{ u, v, w \}$ both $\costs(u, v, w)$ and $\lps(u, v, w)$ depend only on the signs of the edges and the Sherali-Adams solution induced by $\{ u, v, w \}$ (i.e., $y_{uvw}, y_{u|vw}, y_{uv|w}, y_{v|uw}, y_{u|v|w}$). 

\item $\costi(\cdot)$ and $\lpi(\cdot)$ are even more idealized versions of $\costs(\cdot)$ and $\lps(\cdot)$ in the sense that all $+$edges (instead of just medium $+$edges) are rounded with correlation. Formally, $\costi_u(\cdot)$, $\costi(\cdot)$, $\lpi_u(\cdot)$, $\lpi(\cdot)$ are defined identically to $\costs_u(\cdot)$, $\costs(\cdot)$, $\lps_u(\cdot)$, $\lps(\cdot)$ respectively, additionally assuming that instead of running Line~\ref{line:shortone},~\ref{line:shorttwo},~\ref{line:longone},~\ref{line:longtwo} of Algorithm~\ref{algo}, we let $I_p \leftarrow \{ v \in V \setminus \{ p \} : (v, p)\mbox{ is +} \}$ in Line~\ref{line:Ip}. 
\end{itemize}

The superscript $i$ stands for ideal and $s$ stands for special (short and long) edges. 

\begin{remark}
The primary reason that we round short and long $+$edges separately and differentiate $\costs(\cdot)$, $\lps(\cdot)$ from $\costi(\cdot)$, $\lpi(\cdot)$ is to handle the rounding error $\eps_r$ in Lemma~\ref{assump:rounding}, because it applies to every pair $(u, v)$ participating the correlated rounding and we want the $\costi(\cdot)$ and $\lpi(\cdot)$ values for these pairs (more precisely, the triangle $(p, u, v)$) to be large enough to absorb it. For instance, if $\eps_r = 0$ for some $r$, we could have rounded every $+$edge with correlation and just used $\costi(\cdot), \lpi(\cdot)$. 
\label{remark:short}
\end{remark}

We first analyze $\costi(T) / \lpi(T)$ for all triangles. Let $\eta := \finaleta$ and $\gamma := \finalgamma$. 
Call $++-$ triangle $\{ u, v, w \}$ with $(u, v), (u, w)$ being $+$ {\em bad} if $x_{uv}, x_{uw} \in [1/2 - \eta, 1/2 + \eta]$ and $x_{vw} > 1-\eta$. 
The proof of the following lemma appears in Section~\ref{sec:ratio_ideal}.

\begin{lemma}
For any triangle $T$, $\costi(T) / \lpi(T)$ is bounded as follows.

\begin{minipage}{8cm}
\begin{tabular}{|c|c|}
\hline
Type of $T$ & Upper bound  \\
\hline
$+++$ & $\ppprat$ \\
\hline
$+--$ & $\pmmrat$ \\
\hline
$---$ & $\mmmrat$ \\
\hline
$++-$: bad & $\ppmbad$ \\
\hline
$++-$: not bad & $2 - \gamma$ \\% $\ppmnotbad$
\hline
degenerate & $\degrat$ \\
\hline
\end{tabular}
\end{minipage}
\label{lem:triangle_ideal}
\end{lemma}

\noindent Incorporating short and long $+$edges yields the following bounds whose proofs appear in Section~\ref{sec:ratio_short}.
\begin{lemma}
For any triangle $T$, $\costs(T) / \lps(T)$ is bounded as follows.

\begin{minipage}{8cm}
\begin{tabular}{|c|c|}
\hline
Type of $T$ & Upper bound  \\
\hline
$+++$ & $\max(2 - \delta, 1+0.5/(1 - \delta))\leq  \ppprats$ \\
\hline
$+--$ & $\pmmrat(1 + \delta) \leq \pmmrats$\\
\hline
$---$ & $1$ \\
\hline
$++-$: bad & $2$ \\
\hline
$++-$: not bad & $2 - \gamma$ \\
\hline
degenerate & $\degrat$ \\
\hline
\end{tabular}
\end{minipage}

\label{lem:triangle_short}
\end{lemma}

\subsection{Handling Bad Triangles}
\label{sec:charging}
By Lemma~\ref{lem:triangle_short}, the only triangles whose ratio is greater than $2 - \gamma$ are bad triangles; $++-$ triangles with LP distances $x,y,z$ such that $x, y \in [0.5 - \eta, 0.5 + \eta]$ and $z \in (1 - \eta, 1]$. ($x, y$ are $+$edges and $z$ is a $-$edge.) 
Each bad triangle has a unique {\em center}, which is the vertex incident on the two $+$edges. 
Let $\calT$ be the set of all non-degenerate triangles, and $\calD$ be the set of all degenerate triangles. 
In this subsection, given a parameter $\tau > 0$, we will define the charging function $h_\tau : \calT \cup \calD \to \R$ such that 
\begin{itemize}
\item $h_\tau(T) = -\tau$ for every bad triangle $T$.
\item $h_\tau(T) \leq +3\tau$ for {\em chargeable} $T$ which will be defined soon. 
\item $h_\tau(T) = 0$ for all other triangles.
  \item $\sum_{T \in \calT \cup \calD} h_\tau(T) \geq 0$. 
\end{itemize}

For any $p$, let $V_p = \{ u : (u, p) \mbox{ is + and } x_{pu} \in [0.5-\eta, 0.5+\eta] \}$. Every bad triangle $(p, u, v)$ centered at $p$ has $u, v \in V_p$. Consider a graph $G_p = (V_p, E_p)$ whose vertex set is $V_p$ and $(u, v)$ is an edge if and only if $(p,u,v)$ is a bad triangle centered at $p$; in particular, $x_{uv} > 1 - \eta$. 
We prove the following claim that if $(p, u, v)$ and $(p, v, w)$ are bad triangles, then $(p, u, w)$ cannot be bad. 

\begin{claim}
Suppose that $(u, v), (v, w) \in E_p$. Then $x_{uw} \leq 0.5 + 5\eta$.
\label{claim:badbad}
\end{claim}

\begin{proof}
Consider the local distribution on $\{ p, u, v, w \}$ and let 
\begin{itemize}
\item $q_0 = y_{p|uvw} + y_{p|u|vw} + y_{p|uv|w} + y_{p|uw|v} + y_{p|u|v|w}$ (i.e., the probability that $p$ does not belong to the same cluster with any of $u, v, w$).
\item $q_u = y_{pu|vw} + y_{pu|v|w}$ (i.e., the probability that $p$ belongs to the same cluster with only $u$).
\item $q_v = y_{pv|uw} + y_{pv|u|w}$ (i.e., the probability that $p$ belongs to the same cluster with only $v$).
\item $q_w = y_{pw|uv} + y_{pw|u|v}$ (i.e., the probability that $p$ belongs to the same cluster with only $w$).
\item $q_{uv} = y_{puv|w}$.
\item $q_{uw} = y_{puw|v}$.
\item $q_{vw} = y_{pvw|u}$.
\item $q_{uvw} = y_{puvw}$. 
\end{itemize}
Then we have
\begin{align}
& q_0 + q_u + q_v + q_w + q_{uv} + q_{uw}  + q_{vw} + q_{uvw} = 1 \label{eq:bi1} \\
& 1 - x_{pu} = q_u + q_{uv} + q_{uw} + q_{uvw} \in [1/2 - \eta, 1/2 + \eta] \label{eq:biu} \\
& 1 - x_{pv} = q_v + q_{uv} + q_{vw} + q_{uvw} \in [1/2 - \eta, 1/2 + \eta] \label{eq:biv} \\
& 1 - x_{pw} = q_w + q_{uw} + q_{vw} + q_{uvw} \in [1/2 - \eta, 1/2 + \eta] \label{eq:biw} \\
& \eta \geq 1 - x_{uv} \geq q_{uv} + q_{uvw} \label{eq:biuv} \\
& \eta \geq 1 - x_{vw} \geq q_{vw} + q_{uvw}. \label{eq:bivw}
\end{align}
By~\eqref{eq:biu} and~\eqref{eq:biuv}, $q_{u} + q_{uw} \geq 1/2 - 2\eta$.
By~\eqref{eq:biw} and~\eqref{eq:bivw}, $q_{w} + q_{uw} \geq 1/2 - 2\eta$.
Adding these two inequalities and~\eqref{eq:biv} implies 
\[
q_u + q_w + 2q_{uw} + (q_v + q_{uv} + q_{vw} + q_{uvw}) \geq 3/2 - 5\eta. 
\]
Subtracting~\eqref{eq:bi1} from the above implies that $q_{uw} \geq 1/2 - 5\eta$. 
\end{proof}

As an example, note that if $x_{pu} = x_{pv} = x_{pw} = 0.5$ and $x_{uv} = x_{vw} = 1$, then both $u$ and $w$ belong to $p$'s cluster simultaneously if and only if $v$ does not belong to it, which implies that $y_{uw} \geq y_{puw} = 0.5$.

Call $(p, u, v)$ a {\em chargeable} non-degenerate triangle centered
at $p$ if $u, v \in V_p$ and $x_{uv} \leq 1/2 + 5\eta$.  (This is
irrespective of the sign of edge $(u,v)$.)  Note that the definition
$\eta = \finaleta$ ensures $1/2 + 5\eta \leq 1 - \eta$, which means
that no triangle can be both chargeable and bad.  Also, call any
$+$edge $(p, u)$ with $x_{pu} \in [1/2-\eta, 1/2+\eta]$ a chargeable
degenerate triangle or chargeable edge. It is centered at both $p$ and
$u$.  Using the fact that $(u, v), (v, w) \in E_p$ implies that $(p,
u, w)$ is a chargeable triangle, one can prove the following claim.

\begin{claim}
For any $p$, the number of bad triangles centered at $p$ is at most the number of chargeable triangles (non-degenerate and degenerate combined) centered at $p$. 
\label{claim:numbad}
\end{claim}

\begin{proof}
The number of bad triangles centered at $p$ is $|E_p|$, the number of chargeable edges centered at $p$ is $|V_p|$, and the number of chargeable non-degenerate triangles centered at $p$ is the number of pairs $(u, w) \in \binom{V_p}{2}$ such that $u \neq w$ and $(u, v), (v, w) \in E_p$ for some $v \in V_p$ (which implies that $(u, w) \notin E_p$ since $1/2 + 5\eta \leq 1 - \eta$). Let $F_p$ denote the set of such pairs. Note that $E_p$ and $F_p$ are disjoint.

Fix $u \in V_p$ and consider the BFS tree on $G_p$ starting from $u$. With the root being at the zeroth level, the vertices $v$ such that $(u, v) \in F_p$ are exactly the vertices at the second level of the BFS tree. Since there is no triangle in $E_p$, the number of vertices in the second level is at least $\max_{w \in N(u)} \deg(w) - 1$, where $N(u)$ denotes the neighbors of $u$ in $G_p$. So, 
\[
|F_p| \geq \frac{1}{2} \sum_{u \in V_p} (\max_{w \in N(u)} \deg(w) - 1)
\]
and the total number of chargeable triangles is
\[
|F_p| + |V_p| \geq \frac{1}{2} \sum_{u \in V_p} \max_{w \in N(u)} \deg(w).
\]

We finally prove that 
\[
\sum_{u \in V_p} \max_{w \in N(u)} \deg(w) \geq
\sum_{u \in V_p} \deg(u) = 2|E_p|,
\]
which finishes the proof the claim. 
Let $U = \{ u \in V_p : \deg(u) > \max_{w \in N(u)} \deg(w) \}$. Note that $U$ is an independent set in $G_p$. 

We would like to show that there exists a matching between $U$ and $V \setminus U$ saturating $U$. In order to see it, for any $U' \subseteq U$, let $V' = \cup_{u \in U'} N(u)$ and $G'$ be the bipartite graph with vertex set $U' \cup V'$ and the edge set $E_p \cap (U' \times V')$. Let $\deg'(\cdot)$ denote the degree in $G'$, and note that for $(u, v) \in E'$, $\deg'(u) = \deg(u) > \deg(v) \geq \deg'(v)$ by construction. 
Without loss of generality, let $U' = \{ u_1, \dots, u_k \}$ and $V' = \{ v_1, \dots, v_{\ell} \}$ with $\deg'(u_1) \leq \dots \leq \deg(u'_k)$ and $\deg'(v_1) \leq  \dots \leq \deg'(v_{\ell})$. If $\ell < k$, since $|E'| = \sum_{i=1}^k \deg'(u_i) = \sum_{i=1}^{\ell} \deg'(v_i)$, there exists $t \in [\ell]$ such that $\deg'(u_t) \leq \deg'(v_t)$ and $\deg'(u_i) > \deg'(v_i)$ for $i = 1, \dots, t - 1$. However, note that all edges from $u_1, \dots, u_t$ go to $v_1, \dots, v_{t-1}$ while $\sum_{i=1}^t \deg'(u_i) > \sum_{i=1}^{t-1} \deg'(v_i)$, which is contradiction. Therefore, $|V'| \geq |U'|$ for all $U'$, and by Hall's condition, there exists a matching between $U$ and $V_p \setminus U$ saturating $U$. 

Let $(u_1, v_1), \dots, (u_k, v_k)$ be such a matching where $|U| = k$. Let $V' = \{ v_1, \dots, v_k \}$. Note that $\max_{w \in N(v)} \deg(w) \geq \deg(v)$ for every $v \notin U$. Therefore,
\[
\sum_{u \in V_p} \max_{w \in N(u)} \deg(w) 
\geq
\sum_{v \in V_p \setminus (U \cup V')} \deg(v) 
+ \sum_{i=1}^k \big( (\max_{w \in N(u_i)} \deg(w)) +  (\max_{w \in N(v_i)} \deg(w)) \big)
\geq
\sum_{u \in V_p} \deg(u),
\]
which finishes the proof.
\end{proof}

So, around each center $p$, we can let $h_\tau(T) = -\tau$ for every bad triangle $T$ and increase $h_\tau(T)$ by $\tau$ for every chargeable triangle $T$. Chargeable non-degenerate triangles are increased at most three times, and chargeable degenerate triangles are increased at most twice, so $h_\tau(T) \leq 3\tau$ for every $T$.

\subsection{Incorporating Error from Correlated Rounding}
Recall that $\costr(\cdot), \lpr(\cdot)$ are defined with respect to actual rounding, 
and our goal is to bound the ratio $ALG / LP$ where 
\[
n\cdot ALG = \sum_{u \in V} \sum_{(v, w) \in \binom{V}{2}} \costr_u(v, w)
= \sum_{u, v, w \in \binom{V}{3}} \costr(w, u, v) + \sum_{u, v \in \binom{V}{2}} \costr(u, v),
\]
and 
\[
n\cdot LP = \sum_{u \in V} \sum_{(v, w) \in \binom{V}{2}} \lpr_u(v, w)
= \sum_{u, v, w \in \binom{V}{3}} \lpr(w, u, v) + \sum_{u, v \in \binom{V}{2}} \lpr(u, v).
\]

Call a (non-degenerate) triangle $(a, b, c)$ {\em rounded with correlation} when one vertex was the pivot, both of the other vertices are rounded in a correlated manner. 
Let $\calR$ be the set of triangles rounded with correlation. Note that a triangle is in $\calR$ if only if it has at least two medium $+$edges. 
We prove the following claim that relates $ALG, LP$ defined using $\costr, \lpr$ to $\costs, \lps$. 

\begin{claim}
\[
\frac{ALG}{LP} 
\leq \frac{\sum_{u, v, w \in \binom{V}{3}} \costs(w, u, v) + \sum_{T \in \calR} O(\eps_r) + \sum_{u, v \in \binom{V}{2}} \costs(u, v)}{\sum_{u, v, w \in \binom{V}{3}} \lps(w, u, v) - \sum_{T \in \calR} O(\eps_r)  + \sum_{u, v \in \binom{V}{2}} \lps(u, v)}.
\]
\label{claim:realworld}
\end{claim}

\begin{proof}
Note that $\costs(w, u, v) = \costr(w, u, v)$ and $\lps(w, u, v) = \lpr(w, u, v)$ if $(w, u, v) \notin \cal{R}$, so we only need to worry about triangles rounded with correlation.

Fix a pivot $p$ and let $S' \subseteq I_p$ be the random set actually sampled by the algorithm. Then for any $(p, u, v)$ with $u, v \in I_p$, changing from $\costr_p(u, v)$ to $\costs_p(u, v)$ increases the total $ALG$ by at most $3 |\Pr[u, v \in S' \mid p \mbox{ pivot}] - y_{puv}|$. Similarly, changing from $\lpr_p(u, v)$ to $\lps_p(u,v)$ decreases the total LP by at most $3|\Pr[u, v \in S' \mid p \mbox{ pivot}] - y_{puv}|$.

Lemma~\ref{assump:rounding} guarantees that 
\[
\sum_{u, v \in I_p} |\Pr[u, v \in S'] - y_{puv}|
\leq |I_p|^2 \cdot \eps_r,
\]
and there are $\Omega(|I_p|^2)$ triangles of the form $(p, u, v)$ with $u, v \in I_p$. Therefore, 
\[
\sum_{u, v \in I_p} \costr_p(u, v) \leq 
\sum_{u, v \in I_p} (\costs_p(u, v) + O(\eps_r))
\]
and 
\[
\sum_{u, v \in I_p} \lpr_p(u, v) \geq 
\sum_{u, v \in I_p} (\lps_p(u, v) - O(\eps_r)),
\]
so converting $\costr_p, \lpr_p$ to $\costs_p, \lps_p$ for all triangles $\{ (p, u, v) : u, v \in I_p \}$ only increases the $ALG/LP$ ratio if we increase $\costs_p(u, v)$ and decrease $\lps_p(u, v)$ for all these triangles by $O(\eps_r)$. 
Do such a conversion from $\costr,\lpr$ to $\costs,\lps$ for every pivot. 
\end{proof}

\subsection{Finishing Off}
\label{subsec:finish}
We are finally ready to bound $ALG/LP$. 
\begin{lemma}
$ALG/LP \leq \ratio + O(\eps_r)$. 
\label{lem:final}
\end{lemma}
\begin{proof}
By Claim~\ref{claim:realworld}, it suffices to bound 
\[
\frac{\sum_{u, v, w \in \binom{V}{3}} \costs(w, u, v) + \sum_{T \in \calR} O(\eps_r) + \sum_{u, v \in \binom{V}{2}} \costs(u, v)}{\sum_{u, v, w \in \binom{V}{3}} \lps(w, u, v) - \sum_{T \in \calR} O(\eps_r)  + \sum_{u, v \in \binom{V}{2}} \lps(u, v)}.
\]
Recall that
the only triangles whose $\costs/\lps$ ratio is greater than $2 - \gamma$ are bad triangles (i.e., $++-$ triangle with LP value $x,y,z$ such that $x, y \in [0.5 - \eta, 0.5 + \eta]$ and $z \in (1 - \eta, 1]$). Let $\tau > 0$ to be determined, and consider $h_\tau : \calT \cup \calD \to \R$ constructed in Section~\ref{sec:charging}. Since $\sum_T h_\tau(T) \geq 0$, adding $h_\tau(T)$ to the numerator only increases the ratio. 
Let $\calB$ be the set of bad triangles, and $\calC$ be the set of chargeable triangles (non-degenerate and degenerate). Recall that $h_\tau(T) = -\tau$ for $T \in \calB$, $h_\tau(T) \leq 3\tau$ for $T \in \calC$ and $0$ for other triangles. 
Then, the final ratio can be upper bounded by $(\sum_{T \in \calT \cup \calD} \costf(T)) / (\sum_{T \in \calT \cup \calD} \lpf(T))$ where $\1$ denotes the indicator function and 
\begin{itemize}
\item $\costf(T) := \costs(T) + \1[T \in \calR] O(\eps_r) - \1[T \in \calB] \tau + \1[T \in \calC] 3\tau$.
\item $\lpf(T) := \lps(T) - \1[T \in \calR] O(\eps_r)$. 
\end{itemize}

Now we finish by analyzing each triangle individually. Recall that $\calB$ and $\calC$ are disjoint (as discussed in Section \ref{sec:charging}).  We prove the upper bound and lower bounds for $\lps(T)$ when $T$ is bad, chargeable, or rounded with correlation. 
\begin{claim}
For any $T \in \calB \cup \calC$, $\lps(T) \geq \zeta_l := 2 \cdot (1/2 - \eta)^2$. 
For any $T \in \calB$, $\lps(T) \leq \zeta_u := 2(1/2+2 \eta)(1/2 + \eta) + \eta$. 
For any $T \in \calR$, $\lps(T) \geq 2\delta^2$. 
\label{claim:upperlower}
\end{claim}

We also show the upper bound and lower bounds for $\lps(T)$ when $T$ is bad or chargeable, proving Claim~\ref{claim:upperlower}.
\begin{proof}[Proof of Claim \ref{claim:upperlower}]
For the first claim, consider $T \in \calB \cup \calC$ and assume $T$ is non-degenerate. It means that $T = \{ a, b, c \}$ with two $+$edges $(a, b)$ and $(a, c)$ with $x_{ab}, x_{ac} \in [1/2 - \eta, 1/2 + \eta]$.
When $b$ is the pivot, $a$ belongs to $b$'s cluster with probability at least $1/2 - \eta$, and in that case, $(a, c)$, who was contributing at least $1/2 - \eta$ to LP, is removed from the graph.  
One can apply the same argument when $c$ is the pivot is the pivot to ensure that $\lps(T) \geq 2 \cdot (1/2 - \eta)^2$. 
If $T = \{ a, b\}$ is degenerate, $\lps(T) \geq 2 lp(a, b) = 2 (1/2 - \eta)$. 

For the second claim, let $T = \{ a, b, c \}$ be a bad triangle with two $+$edges $(a, b)$ and $(a, c)$ with $x_{ab}, x_{ac} \in [1/2 - \eta, 1/2 + \eta]$ and $x_{bc} \in (1-\eta, 1]$. Then when $b$ is the pivot, the edge $(a, c)$ will be removed when $a$ or $c$ belongs to the same cluster with $b$, which happens with probability at most $y_{ba} + y_{bc} \leq 1/2 + 2\eta$. The case for $c$ is symmetric, so even assuming that $(b, c)$ is always removed when $a$ is the pivot, $\lps(T) \leq 2\cdot (1/2+2\eta)(1/2+\eta) + \eta$. 

For the third claim, consider $T \in \calR$. It means that $T = \{ a, b, c \}$ with two $+$edges $(a, b)$ and $(a, c)$ with $x_{ab}, x_{ac} \in [\delta, 1 - \delta]$. When $b$ is the pivot, $a$ belongs to $b$'s cluster with probability at least $\delta$, and in that case, $(a, c)$, who was contributing at least $\delta$ to LP, is removed from the graph.   One can apply the same argument when $c$ is the pivot to ensure that $\lps(T) \geq 2 \delta^2$.
\end{proof}

%%%%%%%%%%%%%%%%%%%%%%%%

Note that $2\delta^2 = 0.02$ and with $\eta := \finaleta$, we have $0.8612 \geq \zeta_u \geq \zeta_l \geq 0.3472$. 
We finally compute the ratio for each type of triangle. Note than whenever additive $O(\eps_r)$ is applied, we make sure that the denominator $\lps(T)$ is at least some absolute constant. 

\begin{itemize}
\item For $T \in \calB$: Since $\lps(T) \geq \zeta_l$, $\costf(T) / \lpf(T) = (\costs(T) + O(\eps_r) - \tau) /(\lps(T) - O(\eps_r)) \leq 2 + O(\eps_r) - \tau / \zeta_u$. 

\item For $T \in \calC$: Since $\lps(T) \geq \zeta_l$ too, $\costf(T) / \lpf(T) \leq 2 - \gamma + O(\eps_r) + 3\tau / \zeta_l$. 

\item For $T \in \calR \setminus (\calC \cup \calB)$: Since $\lps(T) \geq 2\delta^2$, so $\costf(T) / \lpf(T) \leq 2 - \gamma + O(\eps_r)$.

\item For all other $T$: $\costf(T)/\lpf(T) = \costs(T) / \lps(T) \leq 2 - \gamma$. 
\end{itemize}
The maximum ratio is $\max(2 - \tau / \zeta_u, 2 - \gamma + 3\tau / \zeta_l) + O(\eps_r)$. 
Setting $\tau$ such that 
\[
2 - \tau/0.8612 = 2 - \gamma + 3\tau / 0.3472 \quad \Rightarrow \quad \tau = \frac{\gamma}{1/0.8612 + 3/0.3472} \approx 0.0055,
\]
the final approximation ratio is $2 - \tau / 0.8612 + O(\eps_r) \leq 1.994 + O(\eps_r)$. 
\end{proof}

\section{Bounds for $\costi(\cdot)/\lpi(\cdot)$}
\label{sec:ratio_ideal}

In this section, we bound $\costi(\cdot)/\lpi(\cdot)$, proving Lemma~\ref{lem:triangle_ideal}.
Throughout this section, we consider a fixed triangle $T$ with vertex set
$\{a,b,c\}$ and edge set $(ab,ac,bc)$.  For each type of triangle, we
compute the worst case ratio for $\costi(T)/\lpi(T)$. 
For the sake of brevity, in this section, let $cost(\cdot) := \costi(\cdot)$ and $lp(\cdot) := \lpi(\cdot)$. 
We assume $\frac{0}{0} = 0$.

To compute $cost(T)/lp(T)$, we have $cost(T) = cost_a(bc)
  + cost_b(ac) + cost_c(ab)$ and $lp(T) = lp_a(bc) + lp_b(ac) +
  lp_c(ab)$.  We use $cost_a(bc)$ to denote the probability that edge
  $bc$ is violated given that $a$ is chosen as a pivot (when triangle
  $abc$ is still intact).  We use $lp_a(bc)$ to denote the probability
  that edge $bc$ is decided (i.e., at least one of $b$ or $c$ is
  chosen to be in the cluster with the pivot $a$) times the
  contribution of edge $bc$ to the LP objective function.

\subsection{$+++$ Triangles}

\begin{lemma}
  For a $+++$ triangle $T$ with vertex set $\{a,b,c\}$ and edge set
  $(ab,ac,bc)$,
  $$\frac{cost(T)}{lp(T)} = \frac{cost_a(bc) + cost_b(ac) +
  cost_c(ab)}{lp_a(bc) + lp_b(ac) + lp_c(ab)} \leq \frac{3}{2 +
  y_{abc} + y_{a|b|c}}.$$
  \end{lemma}

\begin{proof}
We have
\begin{eqnarray*}
cost_a(bc) & = &
(y_{ab} - y_{abc}) + (y_{ac} - y_{abc}), \\ lp_a(bc) & = &
(1-y_{bc})((y_{ab}-y_{abc}) + (y_{ac}-y_{abc}) + (y_{abc})),\\
cost_b(ac) & = &
(y_{ab} - y_{abc}) + (y_{bc} - y_{abc}), \\ lp_b(ac) & = &
(1-y_{ac})((y_{ab}-y_{abc}) + (y_{bc}-y_{abc}) + (y_{abc})),\\
cost_c(ab) & = &
(y_{ac} - y_{abc}) + (y_{bc} - y_{abc}), \\ lp_c(ab) & = &
(1-y_{ab})((y_{ac}-y_{abc}) + (y_{bc}-y_{abc}) + (y_{abc})).
\end{eqnarray*}
We can write
the costs using the following shorthand notation.
$$x = y_{ab|c}, ~y = y_{ac|b}, ~z = y_{bc|a}, ~p = y_{abc} \text{ and } q =
y_{a|b|c}.$$  
We have the
following relations:
\begin{eqnarray*}
y_{ab} = y_{ab|c} + y_{abc} = x+p, \quad y_{ac} = y_{ac|b} + y_{abc} =
y + p, \quad y_{bc} = y_{bc|a} + y_{abc} = z+p.
\end{eqnarray*}
Notice that $x + y + z + p +q = 1$.  Then we have the following.
\begin{eqnarray*}
cost_a(bc) & = & x + y,\\ 
lp_a(bc) & = &
(1-z-p)(x + y + p) ~ = ~ (1-z-p)(1-z-q).
\end{eqnarray*}
The costs for edges $ab$ and $ac$ are analogous.
Thus,
\begin{eqnarray}
\frac{cost(T)}{lp(T)} & = &
  \frac{2(x+y+z)}{
(1-x-p)(1-x -q)+
(1-y-p)(1-y-q)+
(1-z-p)(1-z-q)
} \nonumber\\
& = &
\frac{2(1-p-q)}{
(1+p+q)(1-p-q) +3pq + x^2 + y^2 + z^2
}.
\label{ratio1}
\end{eqnarray}
For fixed $p,q$, the ratio in \eqref{ratio1} is maximized when $x^2
+y^2 +z^2$ is minimized, which occurs when $x = y = z = (1-p-q)/3$.
Therefore, for fixed $p$ and $q$, the ratio in \eqref{ratio1} is at most
\begin{eqnarray*}
\frac{2(1-p-q)}{
(1+p+q)(1-p-q) +3pq + \frac{(1-p-q)^2}{3}
} 
& = &
\frac{2(1-p-q)}{
(1+p+q + \frac{1-p-q}{3})(1-p-q) +3pq 
} \\
& = &
\frac{2(1-p-q)}{
(\frac{4}{3} + \frac{2p}{3} + \frac{2q}{3})(1-p-q) +3pq 
}  \\
& = & 
\frac{3(1-p-q)}{
(2 + p + q)(1-p-q) + \frac{9}{2}pq
}\\
& \leq & \frac{3}{2 + p + q}.
\end{eqnarray*}\end{proof}

\subsection{$---$ Triangles}

\begin{lemma}
  For a $---$ triangle $T$ with vertex set $\{a,b,c\}$ and edge set $(ab,ac,bc)$,
  $$\frac{cost(T)}{lp(T)} = \frac{cost_a(bc) + cost_b(ac) + cost_c(ab)}{lp_a(bc) + lp_b(ac) + lp_c(ab)} \leq 1.$$
  \end{lemma}

\begin{proof}
All edges are $-$edges with costs as follows.
\begin{eqnarray*}
cost_a(bc) & = & (1-\sqrt{x_{ab}})  (1-\sqrt{x_{ac}}),\\
lp_a(bc) & = &
(1-x_{bc}) (1 - \sqrt{x_{ab}}  \sqrt{x_{ac}}),\\
cost_b(ac) & = & (1-\sqrt{x_{ab}})  (1-\sqrt{x_{bc}}),\\
lp_b(ac) & = &
(1-x_{ac}) (1 - \sqrt{x_{bc}}  \sqrt{x_{ab}}),\\
cost_c(ab) & = & (1-\sqrt{x_{ac}})  (1-\sqrt{x_{bc}}),\\
lp_c(ab) & = &
(1-x_{ab}) (1 - \sqrt{x_{bc}}  \sqrt{x_{ac}}).
\end{eqnarray*}
So we have
\begin{eqnarray*}
\frac{cost(T)}{lp(T)} & = &
\frac{(1-\sqrt{x_{ab}})(1-\sqrt{x_{ac}})+
  (1-\sqrt{x_{ab}})(1-\sqrt{x_{bc}})+
(1-\sqrt{x_{ac}})(1-\sqrt{x_{bc}})}
{(1-x_{bc}) (1 - \sqrt{x_{ab}}  \sqrt{x_{ac}})+
(1-x_{ac}) (1 - \sqrt{x_{bc}}  \sqrt{x_{ab}})+
(1-x_{ab}) (1 - \sqrt{x_{bc}}  \sqrt{x_{ac}})}.
\end{eqnarray*}
For ease of notation, let $X = x_{ab}, Y = x_{ac}$ and $Z = x_{bc}$.
Then we have
\begin{eqnarray}
\frac{cost(T)}{lp(T)} & = &
\frac{(1-\sqrt{X})(1-\sqrt{Y})+
  (1-\sqrt{X})(1-\sqrt{Z})+
(1-\sqrt{Y})(1-\sqrt{Z})}
{(1-Z) (1 - \sqrt{X}  \sqrt{Y})+
(1-Y) (1 - \sqrt{Z}  \sqrt{X})+
(1-X) (1 - \sqrt{Z}  \sqrt{Y})}. \label{mmm_ratio}
\end{eqnarray}

We will show that the expression in \eqref{mmm_ratio} is always at
most 1 by showing
that the denominator is always at least as large as the
numerator for any $X,Y,Z \in [0,1]$.  This is equivalent to the
following inequality.
\begin{eqnarray*}
  X + Y + Z + 2\sqrt{XY} + 2\sqrt{XZ} + 2\sqrt{YZ}
  & \leq &
  2\sqrt{X} +
2\sqrt{Y} + 2\sqrt{Z} + Z\sqrt{XY} + Y\sqrt{XZ} + X\sqrt{YZ},
\end{eqnarray*}
which is in turn equivalent to the following inequality.
\begin{eqnarray*}
 (\sqrt{X} + \sqrt{Y} + \sqrt{Z})(\sqrt{X} + \sqrt{Y} + \sqrt{Z}) & \leq &
  (\sqrt{X} + \sqrt{Y} + \sqrt{Z})(2 + \sqrt{XYZ}).  
\end{eqnarray*}
Now it remains to prove that when $X,Y,Z \in [0,1]$, the following
inequality holds.
\begin{eqnarray*}
\sqrt{X} + \sqrt{Y} + \sqrt{Z} & \leq & 2 + \sqrt{XYZ}.
  \end{eqnarray*}
This is true if the following inequality
holds
for all $A,B,C \in [0,1]$.
\begin{eqnarray*}
A+B+C & \leq & 2 + ABC.
  \end{eqnarray*}
To see that this last inequality is true, set $A= 1-\alpha, B =
1-\beta$ and $C = 1-\gamma$ for $\alpha, \beta, \gamma \in [0,1]$.
Then we have
\begin{eqnarray*}
A+B+C ~ = ~ 3-\alpha-\beta -\gamma 
& \leq & 2 + (1-\alpha)(1-\beta)(1-\gamma) ~ \iff\\
1 -\alpha-\beta -\gamma 
& \leq & (1-\alpha)(1-\beta)(1-\gamma)~ \iff \\
\alpha \beta \gamma & \leq &  \alpha \beta + \alpha \gamma + \beta
\gamma.
  \end{eqnarray*}
The last inequality is clearly true for $\alpha, \beta, \gamma \in
[0,1]$. 
\end{proof}

\subsection{$+--$ Triangles}

\begin{lemma}
  For a $+--$ triangle $T$ with vertex set $\{a,b,c\}$ and edge set
  $(ab,ac,bc)$,
  $$\frac{cost(T)}{lp(T)} = \frac{cost_a(bc) + cost_b(ac) + cost_c(ab)}{lp_a(bc) + lp_b(ac) + lp_c(ab)} \leq \pmmrat.$$
  \end{lemma}

\begin{proof}
Since $ab$ is a \pedge, we have
\begin{eqnarray*}
cost_c(ab) & = &
\sqrt{x_{ac}}(1-\sqrt{x_{bc}}) + \sqrt{x_{bc}}(1-\sqrt{x_{ac}}), \\ 
lp_c(ab) & = &
x_{ab}(1 - \sqrt{x_{ac}} \sqrt{x_{bc}}).
\end{eqnarray*}
For $-$edges $ac$ and $bc$, we have
\begin{eqnarray*}
cost_b(ac) & = & (1- x_{ab})(1 - \sqrt{x_{bc}}),\\
lp_b(ac) & = &
(1-x_{ac})(1 - x_{ab} \sqrt{x_{bc}}),
\end{eqnarray*}
\begin{eqnarray*}
cost_a(bc) & = &
(1- x_{ab}) (1 - \sqrt{x_{ac}}),\\
lp_a(bc) & = &
(1-x_{bc})(1 - x_{ab} \sqrt{x_{ac}}).
\end{eqnarray*}
Then
\begin{eqnarray*}
\frac{cost(T)}{lp(T)} 
& = &
  \frac{
\sqrt{x_{ac}}(1-\sqrt{x_{bc}}) + \sqrt{x_{bc}}(1-\sqrt{x_{ac}})
+ (1- x_{ab})(2 - \sqrt{x_{bc}} - \sqrt{x_{ac}})
}{
x_{ab}(1 - \sqrt{x_{ac}} \sqrt{x_{bc}}) +
(1-x_{ac})(1 - x_{ab} \sqrt{x_{bc}}) + 
(1-x_{bc})(1 - x_{ab} \sqrt{x_{ac}})
}\end{eqnarray*}
For ease of notation, let $X = x_{ab}, Y = x_{ac}$ and $Z = x_{bc}$.
Notice that we have triangle inequality on these values (i.e., $X+Y
\geq Z, X+Z \geq Y$ and $Y+Z \geq X$).  Without loss of generality, we
assume $Y \leq Z$.
Then
\begin{eqnarray}
\frac{cost(T)}{lp(T)} 
& = &
  \frac{
\sqrt{Y}(1-\sqrt{Z}) + \sqrt{Z}(1-\sqrt{Y})
+ (1- X)(2 - \sqrt{Z} - \sqrt{Y})
}{
X(1 - \sqrt{Y} \sqrt{Z}) +
(1-Y)(1 - X \sqrt{Z}) + 
(1-Z)(1 - X \sqrt{Y})}\nonumber
  \\
& = &
  \frac{2-
2\sqrt{YZ}
- X(2 - \sqrt{Z} - \sqrt{Y})
}{2 -Y -Z + X + X( 
Y \sqrt{Z}
+
 Z \sqrt{Y}
    - \sqrt{YZ}
-  \sqrt{Z}  
-  \sqrt{Y})
  }\nonumber \\
  & = &
  \frac{2-
2\sqrt{YZ}
+ X( \sqrt{Z} + \sqrt{Y}-2)
}{2 -Y -Z + X(\sqrt{Y} + \sqrt{Z} -1)(\sqrt{YZ}-1) 
}.
\label{cost_T-2} 
\end{eqnarray}

First observe that if $Y=Z=1$, then the ratio is $0/0$.  Thus, we assume that $Y$ and $Z$ are not both equal to 1.  
We consider two cases: i) $\sqrt{Y} + \sqrt{Z} \leq 1$, and ii) $\sqrt{Y} + \sqrt{Z} > 1$.  In case i) we will show that ratio is at most 1.

\begin{claim}
  If $\sqrt{Y} + \sqrt{Z} \leq 1$, then
  \begin{eqnarray*}
  \frac{2-
2\sqrt{YZ}
+ X( \sqrt{Z} + \sqrt{Y}-2)
}{2 -Y -Z + X(\sqrt{Y} + \sqrt{Z} -1)(\sqrt{YZ}-1) 
} & \leq & 1.
\end{eqnarray*}
\end{claim}

\begin{cproof}
  The claim is equivalent to showing the following.
  \begin{eqnarray*}
  2- 2\sqrt{YZ}
+ X( \sqrt{Z} + \sqrt{Y}-2) & \leq &
2 -Y -Z + X(\sqrt{Y} + \sqrt{Z} -1)(\sqrt{YZ}-1).
  \end{eqnarray*}
We can rewrite this as
  \begin{eqnarray*}
  Y + Z
    & \leq &
2\sqrt{YZ} +
X 
 + X(\sqrt{Y} + \sqrt{Z} -1)(\sqrt{YZ}-2).
  \end{eqnarray*}
  Notice that $X(\sqrt{Y} + \sqrt{Z} -1)(\sqrt{YZ}-2) \geq 0$, since both of the last two terms are nonpositive.  Thus, it suffices to show
    \begin{eqnarray*}
  Y + Z
    & \leq &
2\sqrt{YZ} + X.
    \end{eqnarray*}
Since we assume that $Z \geq Y$, we have
    \begin{eqnarray*} 
Y + Z ~ = ~2Y + Z-Y ~ \leq ~   2\sqrt{YZ} + Z-Y \leq
2\sqrt{YZ} + X.
    \end{eqnarray*}
\end{cproof}

Now let us now consider case ii) where $\sqrt{Y} + \sqrt{Z} > 1$. 

\begin{claim}\label{clm:z_is_0}
For all $X,Y,Z \in [0,1]$ with $X,Y,Z$ obeying triangle inequality,
the following ratio
  \begin{eqnarray*}
  \frac{2-
2\sqrt{YZ}-X
+ X( \sqrt{Y} + \sqrt{Z}-1)
}{2 -Y -Z + X(\sqrt{Y} + \sqrt{Z} -1)(\sqrt{YZ}-1)}
  \end{eqnarray*}
attains its maximum value
  when $Z= \min\{1, X+Y\}$.  
\end{claim}

\begin{cproof}
Consider
$X,Y,Z \in [0,1]$ such that $Z < X+Y \leq 1$.  Then we show that we
  can increase $Z$ and decrease $Y$ without decreasing the ratio.  For
  $X,Y,Z \in [0,1]$, let $c = \sqrt{Y} + \sqrt{Z}$.  Notice that $c
  \in (1,2)$.  We can rewrite the ratio in the claim as
\begin{eqnarray*}
     \frac{2- 2(c\sqrt{Y}-Y)
+ X(c-2)}{2 -c^2 + 2(c\sqrt{Y}-Y) + X(c -1)((c\sqrt{Y}-Y)-1)}.
\end{eqnarray*}
The numerator is maximized and the denominator is minimized when $Y$
is minimized.  Thus, we can decrease $Y$ to $Y'$ and increase $Z = (c -
\sqrt{Y})^2$ to $Z' = (c- \sqrt{Y'})^2$ until $Z'=X+Y'$ or $Z' = 1$.\end{cproof}

\begin{claim}
Assuming $Z=1$, the maximum value of the ratio in \eqref{cost_T-2} for
$X, Y \in [0,1]$ and $X+Y \geq 1$ is 1.1184.
\end{claim}

\begin{cproof}
By Claim \ref{clm:z_is_0}, we can set $Z=1$.  Then \eqref{cost_T-2} becomes
\begin{eqnarray*}
  \frac{2-X}
{\sqrt{Y}(1-X) + 1}.
  \end{eqnarray*}
  Since both numerator and denominator are always nonnegative for $X,Y
  \in [0,1]$, the ratio is maximized when $Y$ is minimized, which
  occurs when $Y = 1-X$ (since $X+Y \geq Z=1$).  

Now if $Y = 1-X$, then we have:
\begin{eqnarray}
 f(X) =  \frac{2-X}
{\sqrt{1-X}(1-X) + 1}.\label{lastPMM}
\end{eqnarray}
Taking the derivative of this, we obtain
\begin{eqnarray*}
f'(X) = 0 \quad \iff \quad \sqrt{1-X}(4-X)-2 = 0.
\end{eqnarray*}
This last equation is satisfied when $X = .64470$ and the value of \eqref{lastPMM} for this value of $X$ is at most 1.1184.
\end{cproof}

So if $X+Y \geq 1$, then the lemma holds.  It remains to consider the
case in which $Z = X+Y < 1$.  Recall that $\sqrt{Y} + \sqrt{Z} =
\sqrt{Y} + \sqrt{X+Y} > 1$ also holds.
In this case, observe that the ratio in
\eqref{cost_T-2} is at most
\begin{eqnarray}
  \frac{2-
2\sqrt{Y(X+Y)}
+ X(\sqrt{Y} -1)
}{2 -2Y -X + XY\sqrt{X+Y}-X\sqrt{Y} 
  }.\label{cost_T-22}
  \end{eqnarray}

\begin{claim}
For $X+Y < 1$, $\sqrt{Y} + \sqrt{X+Y} > 1$ and
$X, Y \in (0,1]$, the maximum value of the ratio \eqref{cost_T-22} is 1.5.
\end{claim}

\begin{cproof}
  For each $X \in (0,1]$, we define the following functions.
  \begin{align*}
    U_X(Y) &:= 2- 2\sqrt{Y(X+Y)} + X(\sqrt{Y} -1),\\
    V_X(Y) &:= 2 -2Y -X + XY\sqrt{X+Y}-X\sqrt{Y}. 
  \end{align*}
    To show that $U_X(Y)/V_X(Y) \le 3/2$, we will show
  that the function $F_X(Y) := (3/2) V_X(Y) - U_X(Y)$ is decreasing on
  the relevant domain of $Y$.  Then we can evaluate $F_X(Y)$ for $Y =
  .1$; it is sufficient to show that $F_X(.1) \geq 0$.
  
  To show that $F_X(Y)$ is a decreasing function on the relevant
  interval, we argue that $F'_X(Y) := (3/2) V'_X(Y) - U'_X(Y)
  < 0$.
  We have 
  \begin{align*}
    U_X'(Y) := \frac{\partial U_X}{\partial Y} &= \frac{X}{2\sqrt{Y}} - \frac{X+2Y}{\sqrt{Y(X+Y)}}, \\
    V_X'(Y) := \frac{\partial V_X}{\partial Y} &= X\sqrt{X+Y} - \frac{X}{2\sqrt{Y}} + \frac{XY}{2\sqrt{X+Y}} - 2,
  \end{align*}
  Thus, we have
  \begin{eqnarray*}
     F'_X(Y) = \frac{3}{2} \left(X\sqrt{X+Y} - \frac{X}{2\sqrt{Y}} +
     \frac{XY}{2\sqrt{X+Y}} - 2 \right) -  \frac{X}{2\sqrt{Y}} +
     \frac{X+2Y}{\sqrt{Y(X+Y)}},
     \end{eqnarray*}
    and we want to show $F_X(Y) < 0$ for $Y \in (0,.1]$.  This is
      equivalent to showing
    \[X\left(\frac{3}{2}\sqrt{X+Y} - \frac{5}{4} \frac{1}{\sqrt{Y}} +
    \frac{3}{4}\frac{Y}{\sqrt{X+Y}}\right) + \frac{X+2Y}{\sqrt{Y(X+Y)}} <
    3,\]
    which is equivalent to showing
    \begin{eqnarray*}
X\left(\frac{6 (X+Y) \sqrt{Y}}{4\sqrt{Y(X+Y)}} -
\frac{5\sqrt{X+Y}}{4\sqrt{Y(X+Y)}} + \frac{3Y^{3/2}}{4\sqrt{Y(X+Y)}} +
\frac{4}{4\sqrt{Y(X+Y)}}\right) + \frac{2Y}{\sqrt{Y(X+Y)}} < 3.
    \end{eqnarray*}
    Using the facts that $X+Y < 1$ and $-\sqrt{X+Y} < \sqrt{Y}-1$, we have
\begin{eqnarray*}
\frac{X(6 (X+Y)\sqrt{Y}  - 5\sqrt{X+Y} + 3Y^{3/2} + 4) + 8Y}
{4\sqrt{Y(X+Y)}} < 
\frac{X(6\sqrt{Y}  +5(\sqrt{Y}-1) + 3Y^{3/2} + 4) + 8Y}
{4\sqrt{Y(X+Y)}}.
\end{eqnarray*}
Now, it suffices to show
\begin{eqnarray*}
\frac{X(6\sqrt{Y}  +5(\sqrt{Y}-1) + 3Y^{3/2} + 4) + 8Y}
{4\sqrt{Y(X+Y)}} < 3,
\end{eqnarray*}
which is equivalent to showing
\begin{eqnarray*}
X(11\sqrt{Y} -1 + 3Y^{3/2}) + 8Y
< 12 {\sqrt{Y(X+Y)}}.
\end{eqnarray*}
Equivalently, we want to show for all $X,Y \in [0,1]$,
\begin{eqnarray*}
H(X,Y) := X(11 + 3Y) - \frac{X}{\sqrt Y} + 8\sqrt{Y}
- 12 {\sqrt{X+Y}} < 0.
\end{eqnarray*}
In fact, we will show that for each fixed $Y \in [0,1]$, the
function $H_Y(X) := H(X,Y)$ is convex for $X \in[0,1]$.  Thus, we need to check if
$H(X,Y) < 0$
only for the extreme values of $X$, which are $X=0$ and $X=1-Y$.  In
these cases, we have
$H(0,Y) = 8 \sqrt{Y} - 12 \sqrt{Y} < 0,$ and $$H(1-Y,Y) = 
(1-Y)(11 + 3Y) - \frac{1-Y}{\sqrt{Y}} + 8 \sqrt{Y} -12 = 
-8Y -3Y^2 - \frac{1}{\sqrt{Y}} + 9 \sqrt{Y} -1
.$$
It can be verified that this quantity is always negative for $Y \in [0,1]$.

Now we show that for each fixed $Y \in [0,1]$, the function $H_Y(X) :=
H(Y,X)$ is convex.  
We take the derivative with respect to $X$, which is
\begin{eqnarray*}
  \frac{\partial}{\partial X}\left( X (11 + 3 Y) - \frac{X}{\sqrt{Y}} + 8 \sqrt{Y} - 12 \sqrt{X + Y}\right) =
  -\frac{6}{\sqrt{X + Y}} + 3 Y - \frac{1}{\sqrt{Y}} + 11
\end{eqnarray*}
and the second derivative which is 
\begin{eqnarray*}
  \frac{\partial^2}{\partial^2 X}\left( X (11 + 3 Y) - \frac{X}{\sqrt{Y}} + 8 \sqrt{Y} - 12 \sqrt{X + Y}\right) =
  \frac{3}{(X + Y)^{3/2}}.
\end{eqnarray*}
Thus, since the second derivative is positive for all $Y,X \in [0,1]$,
the function is thus convex with respect to $X$.
\end{cproof}
\end{proof}

\subsection{$++-$ Triangles}

\begin{lemma}
  For a $++-$ triangle $T$ with vertex set $\{a,b,c\}$ and edge set
  $(ab,ac,bc)$,
  $$\frac{cost(T)}{lp(T)} = \frac{cost_a(bc) + cost_b(ac) + cost_c(ab)}{lp_a(bc) + lp_b(ac) + lp_c(ab)} \leq 2.$$
  \end{lemma}

\begin{proof}
Edge $bc$ is the $-$edge, so we have
\begin{eqnarray*}
  cost_a(bc) & = & y_{abc}, \\
  lp_a(bc) & = &
y_{bc} (y_{ab}-y_{abc} + y_{ac}).
\end{eqnarray*}
Since $ab$ and $ac$ are both $+$edges, we have
\begin{eqnarray*}
cost_b(ac) & = &
(1-x_{ab})\sqrt{x_{bc}} + (1-\sqrt{x_{bc}})x_{ab}, \\ 
lp_b(ac) & = &
x_{ac}(1-x_{ab} \sqrt{x_{bc}}),
\end{eqnarray*}
\begin{eqnarray*}
cost_c(ab) & = &
(1-x_{ac})\sqrt{x_{bc}} + (1-\sqrt{x_{bc}})x_{ac}, \\ 
lp_c(ab) & = &
x_{ab}(1-x_{ac}\sqrt{x_{bc}}).
\end{eqnarray*}
We use the following for ease of notation.
\begin{eqnarray*}
X = y_{ab} = x+p, \quad Y = y_{ac} = y+p, \quad Z = y_{bc} = z+p, \quad A = 1-X, \quad B = 1-Y, \quad C = 1-Z.
\end{eqnarray*}

\begin{eqnarray}
  \frac{cost(T)}{lp(T)}
  & = &
\frac{y_{abc} + 
  (1-x_{ab})\sqrt{x_{bc}} + (1-\sqrt{x_{bc}})x_{ab} +
  (1-x_{ac})\sqrt{x_{bc}} + (1-\sqrt{x_{bc}})x_{ac}}{y_{bc} (y_{ab}-y_{abc} + y_{ac}) +
x_{ac}(1-x_{ab} \sqrt{x_{bc}}) + x_{ab}(1-x_{ac}\sqrt{x_{bc}})
}\nonumber\\
& = &
\frac{y_{abc} + 
  y_{ab}\sqrt{1-y_{bc}} + (1-\sqrt{1-y_{bc}})(1-y_{ab}) +
  y_{ac}\sqrt{1-y_{bc}} + (1-\sqrt{1-y_{bc}})(1-y_{ac})}{y_{bc} (y_{ab}-y_{abc} + y_{ac}) +
(1-y_{ac})(1-(1-y_{ab}) \sqrt{1-y_{bc}}) + (1-y_{ab})(1-(1-y_{ac})\sqrt{1-y_{bc}})
}\nonumber \\
& = &
\frac{p + 
  X\sqrt{1-Z} + (1-\sqrt{1-Z})(1-X) +
  Y\sqrt{1-Z} + (1-\sqrt{1-Z})(1-Y)}{Z (X-p + Y) +
(1-Y)(1-(1-X) \sqrt{1-Z}) + (1-X)(1-(1-Y)\sqrt{1-Z})
}\nonumber\\
& = &
\frac{p + 
  X\sqrt{C} + (1-\sqrt{C})A +
  Y\sqrt{C} + (1-\sqrt{C})B}{Z (X-p + Y) +
B(1-A \sqrt{C}) + A(1-B\sqrt{C})
}\nonumber \\
& = &
\frac{p + \sqrt{C}(X + Y - (1-X) - (1-Y)) + A + B}
{Z (X-p + Y) + A +B -2AB \sqrt{C}
}\nonumber\\
& = &
\frac{p + \sqrt{C}(2X + 2Y - 2) + A + B}
{Z (X-p + Y) + A +B -2AB \sqrt{C}
}.\label{ratio2_ppm}
\end{eqnarray}

\begin{claim}\label{clm:ppm_xy2}
The ratio in \eqref{ratio2_ppm} is maximized when $A=B$ (which implies
$X=Y$ and $x=y$).
  \end{claim}

\begin{cproof}
  Fix $W = X+Y$.  Then $A+B = 2-W$. 
Then the ratio in \eqref{ratio2_ppm} is equal to
  \begin{eqnarray*}
\frac{p + \sqrt{1-z-p}(2W - 2) + 2-W}
{(z+p) (W-p) + 2-W -2AB \sqrt{1-z-p}
}.
  \end{eqnarray*}
This ratio is maximized when the denominator is minimized, which
occurs when the term $2AB\sqrt{1-z-p}$ is maximized.  For fixed $A+B = 2-W$, this
occurs when $A=B = \frac{2-W}{2}$.
   \end{cproof}

Then we have
  \begin{eqnarray}
\frac{cost(T)}{lp(T)} & \leq & 
\frac{p + \sqrt{1-z-p}(2W - 2) + 2-W}
{(p+z) (W-p) + 2-W - \frac{(2-W)(2-W) \sqrt{1-z-p}}{2}} \nonumber
\\
%& \leq & \frac{p + \sqrt{1-z-p}(2W - 2) + 2-W}
%{(p+z) (W-p) + 2-W + (2W-2 - \frac{W^2}{2}) \sqrt{1-z-p}}
%\\
& \leq & \frac{p + \sqrt{1-z-p}(2W - 2) + 2-W}
{(p+z) (W-p) + \sqrt{1-z-p}(2W-2 - \frac{W^2}{2}) + 2 -W}
.\label{sqrt_last}
  \end{eqnarray}
Notice that $W = x+y+2p$, $Z = z+p$ and $x+y+z+p \leq 1$.  Let $w=
x+y$.  So $w + z + p \leq 1$.

\begin{claim}\label{clm:ppm_zis0_sqrt}
\eqref{sqrt_last} is maximized when $z=0$.  In other words,
we have
\begin{eqnarray}
\frac{p + \sqrt{1-z-p}(2W - 2) + 2-W}
{(z+p) (W-p ) + 2-W -\frac{(2-W)(2-W) \sqrt{1-z-p}}{2}
} & \leq &
\frac{p + \sqrt{1-p}(2W - 2) + 2-W}
{p(W-p ) + 2-W -\frac{(2-W)(2-W) \sqrt{1-p}}{2}
} \label{sqrt_last_w}
\end{eqnarray}
  \end{claim}
\begin{cproof}
When $W=0$, then we have $w=p=0$.  In this case, both the numerator
and the denominator are zero.  So we can assume that $W > 0$.

Fix $p, w$.  Then $z \in [0,1-p-2w]$.
First consider the case in which $W \geq 1$.
Then as $z$ increases, the numerator decreases and the denominator
increases, so the ratio is maximized when $z=0$.  

Next consider the case in which $0 < W < 1$.  
We want to show that
\begin{eqnarray*}
\frac{2-W + p + \sqrt{1-z-p}(2W - 2)}
{2-W + p(W-p) + z(W-p) + \sqrt{1-z-p}(2W-2 - \frac{W^2}{2})} \leq 
\frac{2-W + p + \sqrt{1-p}(2W - 2)}
{2-W + p(W-p) + \sqrt{1-p}(2W-2 - \frac{W^2}{2})}.
\end{eqnarray*}
Let $F = 2-W+p$ and let $G = 2-W + p(W-p)$.  Notice that $G < F$.
Also, let $H = 2-2W$ and $I = 2-2W+W^2/2$.  Notice that $I > H > 0$ $F
> G > 0$.  
\begin{eqnarray*}
\frac{F + \sqrt{1-z-p}(-H)}
{G + z(W-p) + \sqrt{1-z-p}(-I)} \leq 
\frac{F + \sqrt{1-p}(-H)}
{G + \sqrt{1-p}(-I)}.
\end{eqnarray*}
So this inequality holds iff
\begin{eqnarray*}
FG + F\sqrt{1-p}(-I) + G\sqrt{1-z-p}(-H) +
\sqrt{1-z-p}\sqrt{1-p}(-H)(-I) \leq \\
FG + G\sqrt{1-p}(-H) + F\sqrt{1-z-p}(-I) + 
\sqrt{1-z-p}\sqrt{1-p}(-H)(-I) +
z(W-p)(F + \sqrt{1-p}(-H)),
\end{eqnarray*}
which holds iff
\begin{eqnarray*}
F\sqrt{1-p}(-I) + G\sqrt{1-z-p}(-H) +
\leq 
G\sqrt{1-p}(-H) + F\sqrt{1-z-p}(-I) + 
z(W-p)(F + \sqrt{1-p}(-H)).
\end{eqnarray*}
We prove this in two steps.  The second will be to show that 
$z(W-p)(F + \sqrt{1-p}(-H)) > 0$.  The first will be to show that
\begin{eqnarray*}
F\sqrt{1-p}(-I) + G\sqrt{1-z-p}(-H) +
\leq 
G\sqrt{1-p}(-H) + F\sqrt{1-z-p}(-I).
\end{eqnarray*}
This holds iff
\begin{eqnarray*}
G\sqrt{1-z-p}(-H) -G\sqrt{1-p}(-H) 
\leq 
 F\sqrt{1-z-p}(-I)-F\sqrt{1-p}(-I) \iff \\
G\sqrt{1-p}(H) - G\sqrt{1-z-p}(H)
\leq F\sqrt{1-p}(I) -  
 F\sqrt{1-z-p}(I) \iff \\
GH(\sqrt{1-p} - \sqrt{1-z-p})
\leq FI(\sqrt{1-p} -  
 \sqrt{1-z-p}),
\end{eqnarray*}
which holds because $G < F$ and $H < I$.  
Now we need to show
\begin{eqnarray*}
z(W-p)(F + \sqrt{1-p}(-H)) > 0,
\end{eqnarray*}
which holds iff
\begin{eqnarray*}
2-W+p + \sqrt{1-p}(2W-2)  \geq 0.
\end{eqnarray*}
Since $2W-2 < 0$, we have
\begin{eqnarray*}
2-W+p + \sqrt{1-p}(2W-2)  \geq 
2-W+p + (2W-2) = W+ p \geq 0.
\end{eqnarray*}
Thus, we conclude that we
can set $z=0$ to maximize the ratio. 
\end{cproof}
%Now our goal is to upper bound the following ratio for $x+y+p \leq 1$.
%\begin{eqnarray}
%\frac{p + \sqrt{1-p}(2X + 2Y - 2) + A + B}
%{p (X-p + Y) + A +B -2AB \sqrt{1-p}
%}.\label{last_sqrt}
%\end{eqnarray}

\begin{claim}
For each $p \in [0,1]$, the righthandside of \eqref{sqrt_last_w} is maximized
either when $w= 1-p$ or when $w=0$.
\end{claim}

\begin{cproof}
For fixed $p$, we have the following function of $w$.
\begin{eqnarray*}
f_p(w) & = & \frac{\sqrt{1-p}(2w+4p - 2) + 2-w-p}
{p (w+p) + 2-w-2p - \frac{(2-w-2p)(2-w-2p)\sqrt{1-p}}{2}}\\
& = & 
\frac{
\sqrt{1-p}(4p - 2) + 2 -p + w(2\sqrt{1-p} - 1) }
{2 -2p + p^2  - \sqrt{1-p}
(2-4p +2p^2) 
- \frac{\sqrt{1-p}}{2}
(-4w+w^2 +4pw) +w(p-1)}.
\end{eqnarray*}
We want to show that $f_p(w) \leq \alpha$ for $p, w \geq 0$ and $p+w \leq 1$.  Let $f_p = g_p/h_p$.  Then we want to show that $g_p \leq \alpha \cdot h_p$ for $p \in [0,1]$ and $w \in[0,1-p]$.  Thus, we want to evaluate if the function
$$F_p(w) := \alpha \cdot h_p(w) - g_p(w) \geq 0.$$ Notice that
$$F_p'(w) = \alpha\left(- \sqrt{1-p}(-2 + w + 2p) + p-1 \right) - 2\sqrt{1-p} -1,$$ and
$$F_p''(w) = \alpha\left(- \sqrt{1-p}\right).$$
We conclude that $F_p$ is concave and therefore to find the minimum values of $F_p(w)$ for $w \in [0,1-p]$, we need to evaluate the endpoints on the interval $w \in [0,1-p]$.  
\end{cproof}

\begin{claim}
When $w = 0$
\begin{eqnarray*}
\frac{p + \sqrt{1-p}(2W - 2) + 2-W}
{p (W-p) + 2-W - \frac{(2-W)(2-W) \sqrt{1-p}}{2}
} & \leq &  1.76.
\end{eqnarray*}
\end{claim}

\begin{cproof}
When $w= 0$, then $W = w + 2p = 2p$.  So we have
  \begin{eqnarray*}
\frac{p + \sqrt{1-p}(2W - 2) + 2-W}
{p (W-p) + 2-W - \frac{(2-W)(2-W) \sqrt{1-p}}{2}
} & = & 
\frac{\sqrt{1-p}(4p - 2) + 2 - p}
{p^2 + 2-2p - \frac{(2-2p)(2-2p) \sqrt{1-p}}{2}
}.
\end{eqnarray*}
This function of $p$ is maximized when $p = .71415$ and the ratio is at most 1.7538.
\end{cproof}

\begin{claim}
When $w = 1-p$, we have
\begin{eqnarray*}
\frac{p + \sqrt{1-p}(2W - 2) + 2-W}
{p (W-p) + 2-W - \frac{(2-W)(2-W) \sqrt{1-p}}{2}
} & \leq & 
\frac{2p\sqrt{1-p} + 1}
{1- \frac{(1-p)(1-p) \sqrt{1-p}}{2}
} ~ \leq ~ 2.
\end{eqnarray*}
\end{claim}

\begin{cproof}
We want to show that for $p \in [0,1]$, 
\begin{eqnarray}
f(p) = 
\frac{2p\sqrt{1-p} + 1}
{1- \frac{(1-p)(1-p) \sqrt{1-p}}{2}
} & \leq & 2.\label{p_ratio}
\end{eqnarray}
Let
\begin{eqnarray*}
F(p) = 
2- (1-p)(1-p) \sqrt{1-p}-(2p\sqrt{1-p} + 1) = 1- (1-p)^{2.5}-2p\sqrt{1-p}.
\end{eqnarray*}
Then we want to show that $F(p) \geq 0$ for $p \in [0,1]$.
\begin{eqnarray*}
F'(p) = \frac{2.5(1-p)^2 + 3p-2}{\sqrt{1-p}}.
\end{eqnarray*}
It can be seen that $F'(p) \geq 0$ for $p \in [0,1]$.  Thus, we can conclude that $F$ is an increasing function and we only need to check that $F(0) \geq 0$.
Indeed, we have $F(0) = 0$.
\end{cproof}
\end{proof}

\subsubsection{Ratio for $++-$ triangles that are not bad}

Recall that a $++-$ triangle is {\em bad} if the $+$edges have
distances in $[1/2-\eta, 1/2 + \eta]$ and the $-$edge has distance in
$[1-\eta, 1]$.  Thus, there are two cases in which a $++-$ is {\em not
bad}.  Either i)
at least one $+$edge, say $ab$, has $x_{ab} \in [0,1/2-\eta]$ or 
$x_{ab} \in [1/2 + \eta,1]$, or ii) the $-$edge, say $bc$, has
$x_{bc} \in [0,1-\eta]$.

\begin{lemma}
  For a $++-$ triangle $T$ that is {\em not bad}
with vertex set $\{a,b,c\}$ and edge set $(ab,ac,bc)$, 
  $$\frac{cost(T)}{lp(T)} = \frac{cost_a(bc) + cost_b(ac) + cost_c(ab)}{lp_a(bc) + lp_b(ac) + lp_c(ab)} \leq \ppmnotbad.$$
  \end{lemma}

\begin{proof}
We first consider the case in which $x_{ab} \in [0,1/2-\eta]$ or
$x_{ab} \in [1/2 + \eta, 1]$.  
In other words, $y_{ab} \leq 1/2 - \eta$
or $y_{ab} \geq 1/2 + \eta$.  We have the same ratio as
in \eqref{sqrt_last} with one modification.  Let us assume that
$y_{ab} \leq 1/2 - \eta$.  Then the maximum value of $AB$ is
$$AB \leq \left(\frac{2-W}{2} -\eta\right)\left(\frac{2-W}{2} + \eta\right).$$
Then we have
\begin{eqnarray*}
\frac{cost(T)}{lp(T)} & \leq & \frac{p + \sqrt{1-p}(2W - 2) + 2-W}
{p (W-p) + 2-W - \frac{(2-W-2\eta)(2-W + 2\eta) \sqrt{1-p}}{2}
}.  \end{eqnarray*} 
This does not change much from the earlier
analysis: We want to show that this ratio is at most $\ppmnotbad$, so we
only need to check the case in which $w = 1-p$.  In this case, $W = 1+
p$.  The ratio is at most

\begin{eqnarray*}
\frac{\sqrt{1-p}(2p) + 1}
{1 - \frac{(1-p-2\eta)(1-p + 2\eta) \sqrt{1-p}}{2}
}.
  \end{eqnarray*}
To show that this ratio is at most $\ppmnotbad$, we can show, as before, that it
suffices to check the condition when $p=0$.
When $p=0$, we have
\begin{eqnarray*}
\frac{2}
{2 - (1-2\eta)(1+ 2\eta)} & =  & 
\frac{2}
{1 + 4\eta^2}.
  \end{eqnarray*}
When $\eta = 1/12$, this ratio is at most $1.9459 \leq \ppmnotbad$.

The second case is when 
the $-$edge $bc$ has $x_{bc} \in [0, \eta]$.  In this case, $x_{bc}
\leq 1-\eta$.  Therefore, $1-x_{bc} = y_{bc} = z+p \geq \eta$.  
Because, as we have seen in Claim \ref{clm:ppm_zis0_sqrt}, the ratio
is maximized when $z=0$, we just need to compute the ratio \eqref{p_ratio}
for $p = \eta$.  Recall, the ratio is at most
\begin{eqnarray*}
f(p) = \frac{2p\sqrt{1-p} + 1}
{1- \frac{(1-p)(1-p) \sqrt{1-p}}{2}
}.
\end{eqnarray*}
So $f(1/12) = 1.9399 \leq \ppmnotbad$.
\end{proof}

\subsection{Degenerate Triangles}
Let $\{ u, v \}$ be a degenerate triangle. 
\begin{lemma}
$cost(u, v) / lp(u, v) \leq \degrat$. 
\end{lemma}
\begin{proof}
When $(u, v)$ is $+$, $lp(u,v) = cost(u,v) = 2x_{uv}$, so the ratio is $1$.  When $(u, v)$ is $-$, $lp(u, v) = 2y_{uv}$ always. $cost(u,v) = 2(1 - \sqrt{x_{uv}}) \leq 2(1 - x_{uv}) = 2y_{uv}$, so the ratio is at most $\degrat$. 
\end{proof}

\section{Bounds for $\costs(\cdot)/\lps(\cdot)$}
\label{sec:ratio_short}

In this section we bound $\costs(\cdot)/\lps(\cdot)$, proving Lemma~\ref{lem:triangle_short}. We do the case analyses for different types of triangles. 
\subsection{$+++$ Triangles.} 
Suppose that the vertices are $(a, b, c)$ and the LP values are $(x, y, z)$, where $x = x_{bc}, y = x_{ac}$, and $z = x_{ab}$. For the sake of brevity, for the rest of the proof, we let $cost := \costs(a, b, c)$ and $lp := \lps(a, b, c)$. We do the further case analyses depending on how many edges are short. 
\subsubsection{3 short edges}
Note that all edges are rounded independently for every pivot. 
\[
cost = \dy(1 - \dz) + \dz(1 - \dy) + \dx(1 - \dz) + \dz(1 - \dx) + \dx(1 - \dy) + \dy(1 - \dx).
\]
\[
lp = x(1 - \dy \dz) + y(1 - \dx \dz) + z(1 - \dx \dy).
\]
Without loss of generality, assume $0 \leq x \leq y \leq z \leq \delta$. The triangle inequalities impose additional constraints $z \leq x + y \leq 2y$. We want to show that $T(x, y, z) := cost - (2-\delta)lp \leq 0$. Note that for fixed $y$ and $z$, we have that $T$ is convex in $x$, since the coefficient of $x^2$ is 
\[
\frac{2 - 2z^2/\delta - 2y^2/\delta}{\delta} + (2 - \delta)(zy^2 / \delta^2 + yz^2 / \delta^2) > 0.
\]
Therefore, given $y$ and $z$, $T(x, y, z)$ is maximized when $x$ is smallest or largest possible. So $T(x, y, z) \leq \max(T(z - y, y, z), T(y, y, z))$. 

\begin{itemize}
\item First consider the case $x = y$. Let 
\begin{align*}
& T'(y, z) :=  T(y, y, z) = \frac{y^2}{\delta}(4 - \frac{4z^2}{\delta}) - \frac{2y^4}{\delta^2} + \frac{2z^2}{\delta}
- (2 - \delta)\bigg( 2y (1 - \frac{y^2 z^2}{\delta^2}) + z(1 - \frac{y^4}{\delta^2}) \bigg) \\
=&  -2(2 - \delta)y + \frac{4(1 - z^2/\delta)}{\delta}y^2
+ \frac{2(2 - \delta)z^2}{\delta^2} y^3 + y^4 (-2 + (2-\delta)z)/\delta^2 - (2-\delta)z + 2z^2 / \delta.
\end{align*}
Then 
\begin{align*}
\frac{\partial T'(y, z)}{\partial y} := -2(2-\delta) + \frac{8(1-z^2/\delta)}{\delta}y + \frac{6(2- \delta)z^2}{\delta^2}y^2 + y^3 \cdot 4 (-2 + (2-\delta)z)/\delta^2 ,
\end{align*}
and 
\begin{align*}
\frac{\partial^2 T'(y, z)}{\partial y^2} := \frac{8(1-z^2/\delta)}{\delta} + \frac{12(2- \delta)z^2}{\delta^2}y + y^2 \cdot 12(-2 + (2-\delta)z)/\delta^2 > 0,
\end{align*}
which implies that for fixed $z \leq \delta$, the function $T'(y, z)$ is convex for all $y \in [z/2, z]$, which means that $T'(y, z) \leq \max(T'(z/2, z), T'(z, z))$. 
For 
\[
T'(z, z) = 6z^2 / \delta - 6z^4 / \delta^2 - 3(2-\delta)(z - z^5/\delta^2),
\]
we have 
\[
\frac{\partial T'(z, z)}{\partial z^2} = 12 / \delta - 72z^2 / \delta^2 + 60(2-\delta) z^3/\delta^2,
\]
which is nonnegative for every $0 \leq z \leq \delta \leq 0.1$, so $T'(z, z)$ is convex in $z$ and we have $T'(0, 0) = 0$ and $T'(\delta, \delta) = 6\delta - 6\delta^2 - 3(2-\delta)(\delta - \delta^3) = -3\delta^2 +6\delta^3-3\delta^4< 0$. 

Also for
\[
T'(z/2, z) = \frac{3z^2}{\delta} - \frac{9z^4}{8\delta^2}
- (2-\delta)\bigg( 2z - 3z^5/(16\delta^2) \bigg), 
\]
one can similarly prove that it is convex for $z \in [0, \delta]$ and check $T'(0, 0) = 0$ and $T'(\delta/2, \delta) = -2(2-\delta) \delta + 3\delta - (9/8)\delta^2 + (2-\delta)3\delta^3/16 < 0$. 

\item We now consider the case when $x + y = z$, so that 
\[
cost = (\dx + \dy)(2 - 2\dz) - 2\dx\dy + 2\dz.
\]
\[
lp = x + y  - \frac{xy(x+y)z^2}{\delta^2} - z \dx \dy + z = 
2z  - \frac{xyz^3}{\delta^2} - z \dx \dy
\]
So, 
\[
T(z-y, y, z) = (\dx + \dy)(2 - 2\dz) - 2\dx\dy + 2\dz - (2-\delta)\bigg( 2z  - \frac{xyz^3}{\delta^2} - z \dx \dy \bigg) 
\]
For fixed $x + y = z$ with $z \leq \delta \leq 0.1$, 
\[
(\dx + \dy)(2 - \frac{2z^2}{\delta}) + \frac{(2 - \delta)xyz^3}{\delta^2} = (x + y)^2 \bigg( \frac{2 - 2z^2 / \delta}{\delta} \bigg)  + xy \bigg( \frac{(2 - \delta)z^3}{\delta^2} - \frac{2(2 - 2z^2 / \delta)}{\delta} \bigg)
\]
is maximized when $x = 0, y = z$, since the coefficient of $xy$ in the second expression is strictly negative. 
Therefore, it suffices to check 
\[
T(0, z, z) = \dz(2 - 2\dz)  + 2\dz - (2-\delta) 2z = -2(2-\delta)z + 4\frac{z^2}{\delta} - 2\frac{z^4}{\delta^2}.
\]
Again, this function is convex in the interval $[0, \delta]$ and $T(0, 0, 0) = T(0, \delta, \delta) = 0$. 
\end{itemize}

\subsubsection{2 short/1 medium}
Assume that $z$ is medium. Note that all edges are rounded independently for every pivot. Without loss of generality, assume $0 \leq x \leq y \leq \delta \leq z$.  
\[
cost = \dy(1 - z) + z(1 - \dy) + \dx(1 - z) + z(1 - \dx) + \dx(1 - \dy) + \dy(1 - \dx).
\]
\[
lp = x(1 - \dy z) + y(1 - \dx z) + z(1 - \dx \dy).
\]
The triangle inequalities impose additional constraints $z \leq x + y \leq 2y \leq 2\delta$.
We want to show that $T(x, y, z) := cost - (2-\delta)lp \leq 0$. Note that for fixed $y$ and $z$, we have that $T$ is convex in $x$, since the coefficient of $x^2$ is 
\[
\frac{2 - 2z - 2y^2/\delta}{\delta} + (2 - \delta)(zy^2 / \delta^2 + yz / \delta) 
\geq \frac{2 - 4\delta - 2\delta}{\delta} > 0, 
\]
for $\delta \leq 0.1$.  Therefore, given $y$ and $z$, $T(x, y, z)$ is maximized when $x$ is smallest or largest possible. So $T(x, y, z) \leq \max(T(z - y, y, z), T(y, y, z))$. 

\begin{itemize}
\item Let us first consider 
\begin{align*}
&T'(y, z) :=  T(y, y, z) = \frac{y^2}{\delta}(4 - 4z) - \frac{2y^4}{\delta^2} + 2z
- (2 - \delta)\bigg( 2y (1 - \frac{y^2 z}{\delta}) + z(1 - \frac{y^4}{\delta^2}) \bigg) \\
=& -2(2 - \delta)y + \frac{4(1 - z)}{\delta}y^2
+ \frac{2(2 - \delta)z}{\delta} y^3 + y^4 (-2/\delta^2 + (2-\delta)z/\delta^2) + \delta z.
\end{align*}
Then 
\begin{align*}
\frac{\partial T'(y, z)}{\partial y} := -2(2-\delta) + \frac{8(1-z)}{\delta}y + \frac{6(2- \delta)z}{\delta}y^2 + 4y^3 (-2/\delta^2 + (2-\delta)z/\delta^2),
\end{align*}
and 
\begin{align*}
\frac{\partial^2 T'(y, z)}{\partial y^2} := \frac{8(1-z)}{\delta} + \frac{12(2- \delta)z}{\delta}y + 12y^2(-2/\delta^2 + (2-\delta)z/\delta^2) > 0,
\end{align*}
which implies that for fixed $z \leq \delta$, the function $T'(y, z)$ is convex for all $y \in [z/2, \delta]$, which means that $T'(y, z) \leq \max(T'(z/2, z), T'(\delta, z))$. 
For 
\begin{align*}
& T'(\delta, z) = 2\delta(1 - z) + 2\delta(1 - \delta) + 2z(1 - \delta) - (2 - \delta)\bigg( 2\delta(1 - \delta z) + z(1 - \delta^2) \bigg) \\
=& z(2 - 4\delta - (2 - \delta)(1 - 3 \delta^2)) + 4\delta - 2\delta^2 - 2(2 - \delta)\delta \\
=& z(2 - 4\delta - 2 + 6\delta^2 + \delta - 3\delta^3) = z(-3\delta + 6\delta^2 - 3\delta^3)
\end{align*}
it is maximized when $z = \delta$, and $T'(\delta, \delta) > 0$. 

Also for
\begin{align*}
& T'(z/2, z) = 2z + \frac{z^2}{\delta} - \frac{z^3}{\delta} - \frac{z^4}{8\delta^2}
- (2-\delta)\bigg(  2z - \frac{z^4}{4\delta} - \frac{z^5}{16\delta^2} \bigg) \\
=& z(-2 + 2\delta) + z^2/\delta - z^3/\delta + z^4(-(1/8\delta^2) + (2-\delta)/4\delta) + z^5(2 - \delta)/16\delta^2,
\end{align*}
one can similarly prove that it is convex for $z \in [0, 2\delta]$ and check $T'(0, 0) = 0$ and 
\[
T'(\delta, 2\delta) = -4\delta + 4\delta^2 + 4\delta - 8\delta^2 - 2\delta^2 + 4(2 - \delta)\delta^3 + 2(2-\delta)\delta^3 < 0.
\]

\item We now consider the case when $x + y = z$, so that 
\[
cost = (\dx + \dy)(2 - 2z) - 2\dx\dy + 2z.
\]
\[
lp = x + y  - \frac{xy(x+y)z}{\delta} - z \dx \dy + z 
%= 2z  - \frac{xyz^3}{\delta^2} - z \dx \dy
\]
So, 
\[
T(z-y, y, z) = (\dx + \dy)(2 - 2z) - 2\dx\dy + 2z - (2-\delta)\bigg( 2z  - \frac{xyz^2}{\delta} - z \dx \dy \bigg) 
\]
For fixed $x + y = z$ with $z \leq 2\delta \leq 0.3$, 
\[
(\dx + \dy)(2 - 2z) + \frac{(2 - \delta)xyz^2}{\delta} = (x + y)^2 \bigg( \frac{2 - 2z}{\delta} \bigg)  + xy \bigg( \frac{(2 - \delta)z^2}{\delta} - \frac{2(2 - 2z)}{\delta} \bigg)
\]
is maximized when $x = z - \delta, y = \delta$, since the coefficient of $xy$ in the second expression is strictly negative. 
Therefore, it suffices to check 
\[
T(z-\delta , \delta, z) = ((z-\delta)^2/\delta + \delta)(2 - 2z) - 2(z-\delta)^2 + 2z - (2-\delta)\bigg( 2z  - z^2(z-\delta) - z(z-\delta)^2 \bigg).
\]
It is convex in the interval $[\delta, 2\delta]$ with $T'(0, \delta, \delta) = 0$ and $T'(\delta, \delta, 2\delta) < 0$. 
\end{itemize}

\subsubsection{1 short/2 medium}

Let us say $y, z$ are medium. When $a$ is the pivot, the edges $(a, b)$ and $(a, c)$ are rounded with correlation (recall that $y = x_{ac}$ and $z = x_{ab}$). Note that $x = x_{bc} = y_{ab|c} + y_{ac|b} + y_{a|b|c}$.
\begin{align*}
cost &= y_{ab|c} + y_{ac|b} + \dx(1 - z) + z(1 - \dx) + \dx(1 - y) + y(1 - \dx) \\
& \leq x + \dx(1 - z) + z(1 - \dx) + \dx(1 - y) + y(1 - \dx).
\end{align*}
\[
lp = x(1 - y_{a|bc}) + y(1 - \dx z) + z(1 - \dx y).
\]
Then we compose $cost$ and $lp$ into three parts each and bound their ratios.
\begin{itemize}
\item If we consider $x/2 + z(1 - \dx)$ from $cost$ and $z(1 - \dx y)$ from $lp$, then $z(1 - \dx) \leq z(1 - \dx y)$ and $x/2 \leq z/2 = \frac{z(1 - \delta) }{2(1 - \delta)} 
\leq \frac{z(1 - \dx y) }{2(1 - \delta)}$. Therefore, 
\[
x/2 + z(1 - \dx) \leq (1 + \frac{1}{2(1 - \delta)}) z(1 - \dx y).
\]
\item Similarly, 
\[
x/2 + y(1 - \dx) \leq (1 + \frac{1}{2(1 - \delta)}) y(1 - \dx z).
\]
\item Finally, $\dx(1 - z)$ from $cost$ is at most $x(1 - y_{a|bc})$ from $lp$ since $x^2/\delta \leq x$ and $z = y_{a|bc} + y_{a|b|c} + y_{ac|b}$. 
\end{itemize}
Therefore, $cost \leq (1 + \frac{1}{2(1 - \delta)}) lp$. 

\subsubsection{2 long/1 medium, or 3 long}

Assume $z \geq 0.9$. We have
\[
cost = 2x + 2y + 2z - 2xy - 2yz - 2zx
\]
and 
\[
lp = x + y + z - 3xyz
\]

Let $w = x+y$.  Notice that $w \geq
z$ (since $x + y \geq z$).
Consider 
\begin{align*}
& 2 x + 2 y + 2 z - 2 xy - 2 yz - 2 zx - 
1.5(x + y + z - 3 xyz) \\
=& 2z + 2w - 2zw - 2xy - 1.5(z + w - 3zx) = 0.5z + 0.5w - 2zw + xy(-2 + 4.5z).
\end{align*}
For fixed $w$ and $z \geq .9$, it is maximized when $xy$ is maximized,
which occurs when $x$ and $y$ are equal, so we can assume that $x = y
= w/2$.  This yields the expression 
\[
0.5z + 0.5w - 2zw + (-2 + 4.5z)w^2/4
\]
This is linear in $z$, so maximized when $z = 0.9$ or $z = \min(w, 1)$. 
When $z = 0.9$, 
\[
0.45 + 0.5w - 1.8w + (-2 + 4.05)w^2/4
\]
is negative for all $w \in [0.9, 2]$. 
When $z = 1$, 
\[
0.5 + 0.5w - 2w + (-2 + 4.5)w^2/4
\]
is negative for all $w \in [1, 2]$. 
When $z = w$, 
\[
w - 2w^2 + (-2 + 4.5w)w^2/4
\]
is negative for all $w \in [0.9, 1]$.

\subsubsection{1 short/2 long}
Compared to the 1 medium/2 long case, $lp$ increases and $cost$ increases by a factor of at most $1/(1-\delta)$, so the ratio is at most $1.5 / (1 - \delta) \leq 1.6667$. 

\subsubsection{1 short/1 medium/1 long}
Similarly, compared to the 2 medium/1 long case, $lp$ increases and $cost$ increases by a factor of at most $1/(1-\delta)$, so the ratio is at most $1.5 / (1 - \delta) \leq 1.6667$.

\subsubsection{3 medium}

This case is checked in Lemma~\ref{lem:triangle_ideal} and the ratio is at most $\ppprat$. 

\subsubsection{2 medium/1 long}
Let us say $y, z$ are medium. When $a$ is the pivot, the edges $(a, b)$ and $(a, c)$ are rounded with correlation (recall that $y = x_{ac}$ and $z = x_{ab}$). Note that $x = x_{bc} = y_{ab|c} + y_{ac|b} + y_{a|b|c}$.
\begin{align*}
cost &= y_{ab|c} + y_{ac|b} + x(1 - z) + z(1 - x) + x(1 - y) + y(1 - x)
\end{align*}
\[
lp = x(1 - y_{a|bc}) + y(1 - x z) + z(1 - x y).
\]
Focus on the last four terms of $cost$ and the last two terms of $lp$ and consider 
\begin{align*}
& x(1-z)+z(1-x)+x(1-y)+y(1-x) - 1.9 [ y(1 - xz) + z(1 - xy) ] \\
=\, & 2x + (z+y) - 2x(y + z) - 1.9[ y + z - 2xyz].
\end{align*}
For fixed $x$ and $t := y + z$, it is maximized when $y = z = t / 2$, yielding 
\[
2x + t - 2xt - 1.9[ t - xt^2/2] = 
x(2 - 2t + 0.95t^2) - 0.9t.
\]
The coefficient of $x$ is strictly positive, so it is maximized when $x = 1$, so that the expression is at most 
\[
2 - 2.9t + 0.95t^2.
\]
It is negative when $t \geq 1.1$ and at most $0.17$ when $t = 0.9$. (Note that $t \geq x \geq 0.9$.)

For the remaining two terms $y_{ab|c} + y_{ac|b}$ of $cost$ and the first term $x(1 - y_{a|bc})$ of $lp$,
\begin{itemize}
\item $y_{ab|c} + y_{ac|b} \leq \frac{1}{x} \cdot x(1 - y_{a|bc}) \leq 1.9x(1 - y_{a|bc})$. So if $t \geq 1.1$, then the overall ratio is at most $1.9$.
\item If $t \leq 1.1$, then $y_{a|bc} \leq 0.55$ since it contributes to both $y$ and $z$. Therefore, the above inequality $\frac{1}{x}x(1 - y_{a|bc}) \leq 1.9x(1 - y_{a|bc})$ has an additive slack of at least $0.45x(1.9-1/x) \geq 0.3$, which covers the $0.17$ excess. Therefore, the overall ratio is at most $1.9$ in every case. 
\end{itemize}

\subsection{$++-$ Triangles.} Suppose that the vertices are $(a, b, c)$, edge $(a, b)$ is $-$, and the LP values are $(x, y, z)$, where $x = x_{bc}$, $y = x_{ac}$, and $z = x_{ab}$. 
When both + edges are medium, it is handled in Lemma~\ref{lem:triangle_ideal}. Therefore, we only need to handle when either $x$ or $y$ is short or long. Note that in this case, all edges are rounded independently. 

\subsubsection{When a $+$edge is short} 
Assume that $x \leq y$ and $x \leq \delta$.  Let $\alpha = 1.9$ be the targeted ratio. 

\begin{itemize}
\item We first handle the case $x, y \leq \delta$. 
\begin{align*}
cost = & \dy(1 - \sqrt{z}) + \sqrt{z}(1 - \dy) + \dx(1 - \sqrt{z}) + \sqrt{z}(1 - \dx) + (1 - \dx)(1 - \dy) \\
\leq & 2\sqrt{z} + \dy(- 2\sqrt{z}) + \dx(- 2\sqrt{z}) + 1 + \delta^2 =: cost'
\end{align*}
and
\begin{align*}
lp &= x(1 - \dy\sqrt{z}) + y ( 1-\dx \sqrt{z}) + (1 - z)(1 - \dx \dy) \\
&\geq x(1 - \delta \sqrt{z}) + y ( 1- \delta \sqrt{z}) + (1 - z)(1 - \delta^2) =: lp'
\end{align*}
Let $T(x, y, z) = cost' - \alpha lp'$. Since the coefficients of both $z$ and $\sqrt{z}$ are positive, it is maximized when $z = x + y$. 
Given $z = x + y$, $cost$ is maximized when $x^2 + y^2$ is minimized which is the case when $x = y = z/2$. Therefore, 
\[
T(x, y, z) \leq 2\sqrt{z} - \frac{z^{2.5}}{\delta} + 1 + \delta^2 - \alpha (z(1 - \delta \sqrt{z}) + (1-z)(1- \delta^2)),
\]
which is strictly negative for $z \in [0, 2\delta]$. 

\item Assume $x \leq \delta < y$. 
\[
cost = y(1 - \sqrt{z}) + \sqrt{z}(1 - y) + \dx(1 - \sqrt{z}) + \sqrt{z}(1 - \dx) + (1 - \dx)(1 - y)
\]
and 
\[
lp = x(1 - y\sqrt{z}) + y ( 1-\dx \sqrt{z}) + (1 - z)(1 - \dx y).
\]
Let $T(x, y, z) = cost - \alpha lp$. 
The coefficient of $y$ is $1 - \sqrt{z} - \sqrt{z} - (1 - \dx) -
\alpha( -x\sqrt{z} + 1 - \dx \sqrt{z} - \dx(1 - z)) < 0$, which means that $y$
should be minimized.
While decreasing $y$, if it becomes $y = \delta$, then the above case proves the claimed ratio. (We used the fact that $\dy = y$ when $y = \delta$.) 

Then the only other case where we cannot increase $y$ further is when $z = x + y$. For fixed $y$, consider 
$T(x, y, x+y)$ as a function of $x$. 
\begin{align*}
T(x, y, x+y) =& 
1+ \frac{x^2}{\delta}y + 2\sqrt{x+y}(
  1- y - \frac{x^2}{\delta}) \\
 &- \alpha \bigg( 
 1
- \frac{x^2}{\delta}y + (x+y)\frac{x^2}{\delta}y
- xy \sqrt{x+y} 
-y\frac{x^2}{\delta} \sqrt{x+y}
\bigg)
\end{align*}

It is an increasing function in $x$, so the maximum is attained at $x
= \delta$.  To show this, we can take the derivative.  Let $F_y(x) :=
T(x,y,x+y)$.  Recall we have $x+y < 1$.
\begin{eqnarray}
F_y(x) := 1+ \frac{x^2}{\delta}y + 
2\sqrt{x+y}(1- y - \frac{x^2}{\delta})
+ \alpha \left(
-1+ (1 +\sqrt{x+y} - x -y
)\frac{x^2}{\delta}y
+ xy \sqrt{x+y} \right).
\end{eqnarray}

\begin{eqnarray}
F_y'(x) := 
2\frac{x}{\delta}y -
4\sqrt{x+y} \frac{x}{\delta}+
\frac{1}{\sqrt{x+y}}(1- y - \frac{x^2}{\delta})
\\
+ \alpha \left(
 (\frac{1}{2\sqrt{x+y}} - 1)\frac{x^2}{\delta}y
+ (1 +\sqrt{x+y} - x -y)\frac{2x}{\delta}y
+ y \sqrt{x+y} 
+ xy \frac{1}{2\sqrt{x+y}} 
\right).
\end{eqnarray}
For each $y \in [\delta, 1]$ and all $x \in [0,\delta]$, we can show
that $F_y'(x) > 0$ ($F_y(x)$ attains its minimum value of .66 for $x=
.1$ and $y = .18$), which shows that $F_y(x)$ is increasing. 
So we can assume that $x = \delta$, then we have
\begin{eqnarray}
H(y) := 
1+ \delta y + 
2\sqrt{\delta+y}(1- y - \delta)
+ \alpha \left(
-1+ \delta y
 - \delta^2 y
 -\delta y^2
+ 2\delta y \sqrt{\delta+y} \right). 
\end{eqnarray}

\begin{eqnarray}
H'(y) := 
 \delta  -
2\sqrt{\delta+y}
+\frac{1}{\sqrt{\delta+y}}(1- y - \delta)
+ \alpha \left(
\delta 
 - \delta^2
 -2\delta y
+ 2\delta  \sqrt{\delta+y} 
+ \delta y \frac{1}{\sqrt{\delta+y}} 
\right). 
\end{eqnarray}

\begin{eqnarray}
H''(y) := 
- \frac{2}{\sqrt{\delta + y}}
-\frac{1}{2{(\delta+y)}^{3/2}}(1- y - \delta)
+ \alpha \left(
 -2\delta 
+ \frac{2 \delta}{\sqrt{\delta+y}} 
- \delta y \frac{1}{2{(\delta+y)^{3/2}}} 
\right). 
\end{eqnarray}
$H''(y) < 0$ for all $y$,
because $-2+2\delta \alpha < 0$.
So now we set $H'(y) = 0$ and solve for $y$.  
\begin{eqnarray*}
 \delta  -
2\sqrt{\delta+y}
+\frac{1}{\sqrt{\delta+y}}(1- y - \delta)
+ \alpha \left(
\delta 
 - \delta^2
 -2\delta y
+ 2\delta  \sqrt{\delta+y} 
+ \delta y \frac{1}{\sqrt{\delta+y}} 
\right) = 0.
\end{eqnarray*}

This function has one root for $y \in [\delta, 1]$ at $y \approx
.342$.  We can verify that for $x = \delta, y = .342$ and $z = y+x$, the
function $T(x,y,z) \leq 0$.
\end{itemize}

\subsubsection{When a $+$edge is long}  
Now we assume that one $+$edge is long and the other $+$edges is either long or medium. 
Recall that 
\[
cost = (1-x)(1-y) + y (1-\sqrt{z})
  + \sqrt{z} (1-y) + x (1-\sqrt{z}) + \sqrt{z} (1-x)
\]
and
\[
lp = (1-z)(1- x y) + x(1- y \sqrt{z}) +
y(1- x \sqrt{z}),
\]
so that the ratio is 
\begin{eqnarray}
\frac{1 + xy +  2 \sqrt{z}
  - 2 y \sqrt{z} - 2 x \sqrt{z}}
{(1-z)(1- x y) + x(1- y \sqrt{z}) +
y(1- x \sqrt{z}).
}
\end{eqnarray}

So we want to prove the following inequality, assuming triangle inequality on $x,y,z$ and $x
\geq .9$.  
\begin{eqnarray}
\frac{1 + xy + 2 \sqrt{z}
  - 2 y \sqrt{z} - 2 x \sqrt{z}}
{1 - z - xy + xyz + x + y
- 2xy \sqrt{z}}
~ = ~ \frac{1 + xy + 2 \sqrt{z} ( 1 
  - y -  x)}
{1 - z - xy + xyz + x + y
- 2xy \sqrt{z}.
} & \leq & \frac{3}{2}.
\label{newest_ratioppm}
\end{eqnarray}

\begin{eqnarray}
1 + xy + 2 \sqrt{z} ( 1 
  - y -  x) & \leq & \frac{3}{2} \left(1 - z - xy + xyz + x + y
- 2xy \sqrt{z}\right).
\end{eqnarray}

\begin{eqnarray}
-\frac{1}{2} +
 xy + 2 \sqrt{z} ( 1 
  - y -  x + \frac{3}{2}xy) + \frac{3}{2} z( 1 - xy) + \frac{3}{2} xy 
-\frac{3}{2} \left( x + y \right) \leq 0 .
\end{eqnarray}
For $x, y \in [0,1]$, the
coefficients of $z$ and $\sqrt{z}$ are always nonnegative.  Thus, we
can assume that $z= \min\{x+y, 1\}$. 

We consider two cases: $z=1$ and $z = x+y <1$.
Let us first assume that $z=1$. 
We rewrite the ratio as
\begin{eqnarray}
\frac{1 + xy + 2( 1 
  - y -  x)}
{- 2xy + x + y
} = \frac{3 +xy -2x - 2y}
{x+y - 2xy}.
\end{eqnarray}
To upper bound the ratio by $\frac{3}{2}$, it is thus enough to show that
\begin{eqnarray}
3 +xy -2x -2y - \frac{3}{2}(x + y
- 2xy) \le 0.
\end{eqnarray}
Which we can rearrange as
\begin{eqnarray}
3 +xy -2x -2y - \frac{3}{2}x - \frac{3}{2}y + 3xy \leq 0.
\end{eqnarray}

\begin{eqnarray}
3 + 4xy -  \frac{7}{2}x - \frac{7}{2}y \leq 0
\end{eqnarray}
The LHS is linear in $y$ and so maximized for $y=0$ or $y=1$.
It is thus enough to show:
\begin{eqnarray}
3 + 4x  - \frac{7}{2}x - \frac{7}{2} \leq 0,
\end{eqnarray}
which holds for $x \le 1$; and
\begin{eqnarray}
3-  \frac{7}{2}x \leq 0
\end{eqnarray}
which holds for $ x \ge 6/7$, so it also holds when $x \ge .9$.

\vspace{5mm}

Next, we prove the desired upper bound on the ratio
for the case $z = x+y < 1$.
We have
\begin{eqnarray}
\frac{1 + xy + 2 \sqrt{x+y} ( 1 
  - y -  x)}
{1 - xy + xy(x+y)
- 2xy \sqrt{x+y}.
} & \leq & 
\frac{1 + xy  + 2 ( 1 
  - y -  x)}
{1 - xy + xy(x+y)
- 2xy
}. 
\end{eqnarray}
We want to show that this ratio is at most $\frac{3}{2}$.
\begin{eqnarray}
1 + xy  + 2 ( 1 
  - y -  x) \leq \frac{3}{2} (1 - xy + xy(x+y)
- 2xy) \\
3 + xy   - 2y -  2x \leq 
\frac{3}{2} - \frac{3}{2}xy + \frac{3}{2}xy(x+y) - 3xy\\
\frac{3}{2} + \frac{11}{2}xy  \leq 
2y + 2x + \frac{3}{2}xy(x+y).
\end{eqnarray}
So we have
\begin{eqnarray}
F_x(y): = \frac{3}{2} + \frac{11}{2}xy  
-2y -2x - \frac{3}{2}xy(x+y),
\end{eqnarray}
and we want to show that this function is at most 0 when $x \in
[.9,1]$ and $y \in [0,1]$.
Since
$F_x''(y) = -3x,$ the function is always concave in $y$ for any $x$.  
If
\begin{eqnarray}
F_x'(y) = \frac{11}{2}x - 2 - \frac{3}{2}x^2 - 3xy = 0,
\end{eqnarray}
then $y = y^* = \frac{11}{6} - \frac{2}{3x} -\frac{x}{2}$.
However, this value of $y$ is much larger than $.1$, which is the
maximum value of $y$ allowed (i.e., the maximum is outside the
interval $[0,.1]$).  Thus, it suffices to check the extreme values of
$y = 0$ and $y = 1-x$.  When $y = 0$, we have
\begin{eqnarray}
F_x(y): = \frac{3}{2} -2x \leq 0 ,
\end{eqnarray}
when $x \geq 3/4$.  When $y = 1-x$, we have
\begin{eqnarray}
F_x(y): & = &  
\frac{3}{2} + \frac{11}{2}x(1-x)  
-2(1-x) -2x - \frac{3}{2}x(1-x)\\
& = & -\frac{1}{2} - \frac{8}{2}x^2 + \frac{8}{2}x. 
\end{eqnarray}
This is at most 0 when
\begin{eqnarray}
-1 - 8x^2 + 8x \leq 0,
\end{eqnarray}
and it can be verified that this is the case when $x \in [.9,1]$.

\subsection{$+--$ Triangles} Suppose that the vertices are $(a, b, c)$, edge $(b, c)$ is $+$, and the LP values are $(x, y, z)$, where $x = x_{bc}$, $y = x_{ac}$, and $z = x_{ab}$. 

\begin{itemize}
\item $x$ is medium or long: This case is checked in Lemma~\ref{lem:triangle_ideal} and the ratio is at most $\pmmrat$. 

\item $x$ is short: Note that all triangles are rounded independently. 
\[
cost = \sqrt{z}(1-\sqrt{y}) + \sqrt{y}(1-\sqrt{z}) + (1-\dx)(1 - \sqrt{z}) + (1- \dx)(1 - \sqrt{y}),
\]
and
\[
lp = x(1 - \sqrt{y}\sqrt{z}) + (1 - y)(1 - \dx  \sqrt{z}) + (1 - z)(1 - \dx \sqrt{y}). 
\]
Note that the expressions for $cost$ and $lp$ when $x$ is long are identical to the above, except that $\dx$ is replaced by $x$. Since $\dx \leq x \leq \delta$ and $(1 - \dx) \in [1 - x, (1 + \delta)(1 - x)]$, $cost$ for short $x$ is at most $(1+\delta)$ times $cost$ for long $x$, and $lp$ for short $x$ is at least $lp$ for long $x$. Since the ratio $cost / lp$ for medium/long $x$ is at most $\pmmrat$, the ratio for short $x$ is at most $\pmmrat(1 + \delta)$. 
\end{itemize}

\subsection{$---$ Triangles} This case is already checked in Lemma~\ref{lem:triangle_ideal}, and the ratio is at most $\mmmrat$.

\subsection{Degenerate triangles} 
Let $\{ u, v \}$ be a degenerate triangle. 
Compared to degenerate triangles in Lemma~\ref{lem:triangle_ideal}, the only change happens when $(u, v)$ is a short $+$edge, which makes $cost(u, v) = 2x_{uv}^2/\delta \leq 2x_{uv}$. Therefore, the ratio is still at most $\degrat$.

\section{Details of Correlated Rounding}
\label{sec:rt}
In this section, we prove Lemma~\ref{assump:rounding}. 
Recall that given a correlation clustering instance $G = (V, E)$ and a
solution $y$ to the $r$-rounds of Sherali-Adams, we chose a pivot $p
\in V$ and let $I_p$ be the set of vertices that have a medium $+$edge to $p$. We would like to sample a set $S' \subseteq I_p$ such that (1) for each $v \in I_p$, $\Pr[v \in S] = y_{pv}$ and (2) $\E_{u, v \in I_p}[|\Pr[u, v \in S] - y_{puv}|] \leq \eps_r$, where $\eps_r = O(1/\sqrt{r})$.

Note that sampling $S' \subseteq I_p$ is equivalent to making a binary decision for each $v \in I_p$; whether to put $v$ into $S'$ or not. In this interpretation, we can almost directly import the tools for CSPs (with binary alphabets). For sake of completeness, we show how the framework of Raghavendra and Tan~\cite{RT12} used for Max-CSPs with cardinality constraints can be used for our purpose. Similar techniques also have been used for non-constrained CSPs and graph partitioning problems~\cite{GS11, BRS11}. 

Imagine we are interested in a CSP that has $n$ variables $W = \{ v_1, \dots, v_n \}$ where each $v_i$ can have a value in $\{ 0, 1 \}$. (Predicates and objective functions are not important here.) The $r$-rounds of Sherali-Adams for the CSP have variables $x_{S, \alpha}$ for any $S \subseteq W$, $|S| \leq t$ and $\alpha \in \{ 0, 1 \}^S$ (also interpreted as a function $\alpha : S \to \{ 0, 1 \}$) where $x_{S, \alpha}$ denotes the probability that the variables in $S$ are assigned $\alpha$. The following constraints ensure that these {\em local distributions} are consistent. Given $\alpha \in \{ 0, 1\}^S$ and $T \subseteq S$, let $\alpha|_T \in \{ 0, 1 \}^T$ be the restriction of $\alpha$ to $T$. 

\begin{align}
& x_{\emptyset} = 1. \label{eq:sa_csp_one}\\
& x_{T, \beta} = \sum_{\alpha \in \{ 0, 1 \}^S : \alpha|_T = \beta} x_{S, \alpha} && T \subseteq S \subseteq W, |S| \leq r, \beta \in \{ 0, 1 \}^T. \label{eq:sa_csp_two} \\
& x \geq 0. \label{eq:sa_csp_three}
\end{align}

In our setting where $W = I_p$ and $v_i = 1$ indicates that $v_i$ is put into $S'$, it is natural to associate $x_{v_i, 1} = y_{pv_i}$. The following claim shows that such association can be formally defined for higher-level variables as well. 

\begin{claim}
Let $y$ be a solution to the $r$-rounds of Sherali-Adams for \cc, $p \in V$, $W \subseteq V \setminus \{ p \}$ and define $\{ x_{S, \alpha} \}_{S \subseteq W, |S| \leq r - 1, \alpha \in \{ 0, 1 \}^S}$ as 
\[
x_{S, \alpha} = \sum_{\substack{S_1,\ldots, S_{\ell}: \\ S = S_1 \cupdot \ldots \cupdot S_{\ell} \\ \text{and } S_1 = \alpha^{-1}(1)}} y_{S_1 \cup \{ p \}|S_2|\dots|S_{\ell}}. 
\]
Then $x$ satisfies the Sherali-Adams constraints for CSP~\eqref{eq:sa_csp_one},~\eqref{eq:sa_csp_two}, and~\eqref{eq:sa_csp_three} with $(r-1)$ rounds. 
\end{claim}
\begin{proof}
The only nontrivial constraint is~\eqref{eq:sa_csp_two}. Fix $T \subseteq S \subseteq W$ with $|S| \leq r - 1$ and $\beta \in \{ 0, 1 \}^T$. Let $T_1 = \beta^{-1}(1)$ and $T_0 = \beta^{-1}(0)$. 
\begin{align*}
x_{T, \beta} &= \sum_{\substack{T_2,\ldots, T_{\ell}: \\ T_0 = T_2 \cupdot \ldots \cupdot T_{\ell}}} y_{T_1 \cup \{ p \}|T_2|\dots|T_{\ell}} \\
&= \sum_{\substack{T_2,\ldots, T_{\ell}: \\ T_0 = T_2 \cupdot \ldots \cupdot T_{\ell}}}
\sum_{\substack{S_1, \dots, S_{q} : \\ S = S_1 \cupdot  \ldots \cupdot S_{q} \mbox{ and} \\ T_i = S_i \cap T~ \forall i \in [\ell]}}
 y_{S_1 \cup \{ p \}|S_2|\dots|S_{q}} \\
 & = \sum_{S_1 : T_1 = S_1 \cap T}
  \bigg(
 \sum_{\substack{T_2,\ldots, T_{\ell}: \\ T_0 = T_2 \cupdot \ldots \cupdot T_{\ell}}}
\sum_{\substack{S_2, \dots, S_{q} : \\ S = S_1 \cupdot  \ldots \cupdot S_{q} \mbox{ and} \\ T_i = S_i \cap T~ \forall i \in \{ 2, \dots, \ell \} }}
 y_{S_1 \cup \{ p \}|S_2|\dots|S_{q}} \bigg) \\
  & = \sum_{S_1 : T_1 = S_1 \cap T}
  \bigg(
\sum_{\substack{S_2, \dots, S_{q} : \\ S = S_1 \cupdot  \ldots \cupdot S_{q}}}
 y_{S_1 \cup \{ p \}|S_2|\dots|S_{q}} \bigg) \\
      & = \sum_{S_1 : T_1 = S_1 \cap T}
   \sum_{\alpha \in \{ 0 , 1\}^S : \alpha^{-1}(1) = S_1} x_{S, \alpha} 
   = \sum_{\alpha \in \{ 0 , 1\}^S : \alpha|_T = \beta} x_{S, \alpha}.
 \end{align*}
\end{proof}

Therefore, $\{ x_{S, \alpha} \}$ is a valid solution to the $(r-1)$ rounds of the Sherali-Adams hierarchy (for CSPs). Then in order to finish Lemma~\ref{assump:rounding} it suffices to give a randomized rounding algorithm that outputs $0$-$1$ random variables $\{ X_v \}_{v \in W}$ such that $\E[X_v] = x_v = y_{pv}$ for each $v \in W$ and $\E_{u, v \in W} \E[X_u X_v] = \E_{u, v} [x_{uv}] \pm \eps_r  = \E_{u, v}[y_{puv}] \pm \eps_r$. At this point, Theorem 4.6 of Raghavendra and Tan~\cite{RT12} shows that such rounding a exists. For sake of completeness, we reproduce their proof here. 

Their rounding is to (1) carefully choose the seed set $S \subseteq W$ with $|S| \leq r - 2$, (2) round $\{ X_v \}_{v \in S}$ according to the joint distribution $\{ x_{S, \alpha} \}_{\alpha \in \{0,1\}^S}$, and (3) for each $u \in W \setminus S$, independently round $X_u$ from the conditional distribution given the rounded values for $\{ X_v \}_{v \in S}$. (Since $x$ is a solution for $r-1$ rounds and $|S| \leq r - 2$, the conditional rounding is possible.)

\cite{RT12} showed how to find a good seed and analyzed the performance of the rounding using entropy. Recall that for $0$-$1$ random variables $X$ and $Y$, their entropy, mutual entropy, and conditional entropy are defined as 
\begin{align*}
H(X) &:= -\sum_{i \in \{ 0, 1 \}} \Pr[X = i] \log \Pr[X = i], \\
I(X; Y) &:= \sum_{i, j \in \{ 0, 1 \}} \Pr[X = i, Y = j]\log \frac{\Pr[X = i, Y = j]}{\Pr[X = i]\Pr[X = j]}, \\
H(X|Y) &:= \sum_{i \in \{ 0, 1 \}} \Pr[Y = i] H(X | Y = i).\\
\end{align*}
The mutual information and the pairwise correlation can be related as follows. 
\begin{claim} [Fact 4.3 of~\cite{RT12}] 
For any $i, j \in \{ 0, 1 \}$, 
\[
|\Pr[X = i, Y = j] - \Pr[X = i]\Pr[Y = j]| \leq \sqrt{2I(X; Y)}.
\]
When $i = j = 1$, this implies that 
\[
|\E[XY] - \E[X]\E[Y]| \leq \sqrt{2I(X; Y)}.
\]
\end{claim}
For a seed $S$, let $X_S = \{ x_v \}_{v \in S}$. We want to find a
good seed $S$ with $|S| \leq r-2$ such that $\E_{u, v \in W}
\sqrt{2I(X_u ; X_v | X_S)}$ is small.
The following lemma guarantees that there exists a good seed.
\begin{lemma}
There exists $t \leq r - 3$ such that 
\[
\E_{w_1, \dots, w_t \in W} \E_{u, v \in W} [I(X_u ; X_v | X_{w_1}, \dots, X_{w_t})] \leq 1/(r - 2).
\]
\end{lemma}
\begin{proof}
By linearity of expectation, we have that for any $t \leq r - 2$, 
\begin{align*}
&\E_{u, w_1, \dots, w_t}[H(X_u|X_{w_1}, \dots, X_{w_t})]\\
= &
\E_{u, w_1, \dots, w_t}[H(X_u|X_{w_1}, \dots, X_{w_{t-1}})]
-
\E_{w_1, \dots, w_{t-1}} \E_{u, w_t} [I(X_u ; X_{w_t} | X_{w_1}, \dots, X_{w_{t-1}})]
\end{align*}

adding the equalities from $t = 1$ to $t = r - 2$, the lemma follows since
%\[
%1 \geq \E_{u \in W}[H(X_u)]  - \E_{u, w_1, \ldots, w_{r - 2} \in W}[H(X_u)]
%= \sum_{1 \leq t \leq r - 2} \E_{u, v, w_1, \dots, w_{t-1}} [ I(X_u ; X_v | X_{%w_1}, \ldots, X_{w_{t-1}})].
%\]
  \[\hspace{-2mm}
1 \geq \E_{u \in W}[H(X_u)]  - \E_{u, w_1, \ldots, w_{r - 2} \in W}[H(X_u | X_{w_1}, \dots, X_{w_{r-2}})]
= \sum_{1 \leq t \leq r - 2} \E_{u, v, w_1, \dots, w_{t-1}} [ I(X_u ; X_v | X_{w_1}, \ldots, X_{w_{t-1}})].
\]
\end{proof}
Therefore, there exists $S \subseteq W$ with $|S| \leq r - 3$ such that 
$\E_{u,v}[I(X_u ; X_v | X_S)] \leq 1 / (r-2)$. Find such an $S$ by
exhaustive search, sample $\alpha_S \in \{ 0, 1 \}^S$ from the local
distribution 
(i.e., according to the convex combination of
solutions on $S$), 
%(i.e., from $\{ x_{S, \alpha} \}_{\alpha}$),
let $X_S = \alpha_S$, and for each $v \in W \setminus S$, round $X_v$ independently conditioned on $X_S = \alpha_S$. (One can also interpret that $X_v$ for $v \in S$ is also independently rounded again conditioned on $X_S = \alpha_S$, though this rounding does not change anything.) 
Note that for any $v \in W$, the marginal is exactly preserved; i.e., $\E[X_v] = \E_{\alpha_S} \E[X_v | X_S = \alpha_S] = x_v$. Finally,
\begin{align*}
& \E_{u, v \in W} |\E[X_u X_v] - x_{uv}| \\
= & \E_{u, v \in W} |\E_{\alpha_S} (\E[X_u |X_S = \alpha_S] \E[X_v | X_S = \alpha_S] -  \E[X_u X_v | X_S = \alpha_S]) | \\
\leq & \E_{u, v \in W} \E_{\alpha_S} |\E[X_u |X_S = \alpha_S] \E[X_v | X_S = \alpha_S] -  \E[X_u X_v | X_S = \alpha_S] | \\
\leq & \E_{u, v \in W} \E_{\alpha_S} \sqrt{2I(X_u ; X_v | X_S = \alpha_S)} \\
\leq &\E_{u, v \in W} \sqrt{2I(X_u ; X_v | X_S)} \\
\leq & \sqrt{\E_{u, v \in W} 2I(X_u ; X_v | X_S)} \leq O(1/\sqrt{r}). 
\end{align*}

\section{Derandomization}
\label{sec:derandomization}
In this section, we show that our algorithm can be derandomized. Fix one iteration with $G = (V, E)$. For any $p \in V$, Section~\ref{sec:rt} shows how to deterministically find of a good seed set $T_p \subseteq I_p$ with $|T_p| \leq r - 2$. Then the algorithm for one iteration can be abstractly described as follows.
\begin{enumerate}
\item Sample $p \in V$. Recall $I_p = \{ u : (p, u)\mbox{ is medium }+ \}$. Let $S \leftarrow \emptyset$. 
\item For each $u \in V \setminus (I_p \cup \{ p \})$, independently decide $S \leftarrow S \cup \{ u \}$ or not with the probability depending on $x_{pv}$. 
\item Sample $T', S' \subseteq I_p$ as follows.
\begin{itemize}
\item Sample $T' \subseteq T_p$ according to the local distribution of the Sherali-Adams solution induced by $T_p \cup \{ p \}$. 
\item For each $u \in I_p \setminus T_p$, independently decide $S' \leftarrow S' \cup \{ u \}$ or not with the probability according the local distribution of the Sherali-Adams solution induced by $T_p \cup \{ p, u\}$ (conditioned on $T'$). 
\end{itemize}
\item Make $\{ p \} \cup S \cup T' \cup S'$ as a new cluster. 
\end{enumerate}

From the description, it is clear that $\cost_p(u,  v)$ and $\lp_p(u,  v)$ can be deterministically computed in polynomial time; once $T_p$ is given, one can go over each possible $T' \subseteq T_p$ (there are at most $2^{|T_p|} \leq 2^r$ choices), and the rest of the rounding is independent for each vertex. 

Since Lemma~\ref{lem:final} prove that $ALG/LP \leq \ratio + \eps$, there exists $p \in V$ such that 
\[
\frac{\sum_{(u, v) \in E} \costr_p(u, v)}{\sum_{(u, v) \in E} \lpr_p(u, v)}
\]
is at most $\ratio + \eps$, and one can deterministically compute such $p$ since $\cost_p(u, v)$ and $\lp_p(u, v)$ are already computed.

Once $p$ is chosen, for each possible $T' \subseteq T_p$, one can compute the expected value of $\sum_{(u, v) \in E} \costr_p(u, v)$ and $\sum_{(u, v) \in E} \lpr_p(u, v)$ conditioned on $T'$. Find a $T'$ that makes the ratio still $\ratio + \eps$. Conditioned on $T'$, the rest of the rounding is independent for every vertex $V \setminus (T_p \cup \{ p \})$, and one can continue to apply the method of conditional expectations for each vertex to decide whether it belongs to $p$'s cluster or not. At the end, we deterministically compute a cluster including $p$ whose removal incurs the cost of $\alpha$ and decreases the remaining LP value by $\beta$, where $\alpha \leq (\ratio + \eps)\beta$. Iterating this method for every iteration until the end ensures that the total cost is at most $(\ratio + \eps)$ times the original LP value.

\section{Acknowledgements}
We thank Shi Li for pointing out a missing case in Claim
\ref{clm:ppm_zis0_sqrt}.
  
\bibliographystyle{alpha}
\bibliography{refs}

\newcommand{\etalchar}[1]{$^{#1}$}
\begin{thebibliography}{KCMNT08}

\bibitem[AAEG15]{anava2015improved}
Yael Anava, Noa Avigdor-Elgrabli, and Iftah Gamzu.
\newblock Improved theoretical and practical guarantees for chromatic
  correlation clustering.
\newblock In {\em 24th International World Wide Web Conference (WWW)}, 2015.

\bibitem[AALvZ12]{ailon2012improved}
Nir Ailon, Noa {Avigdor-Elgrabli}, Edo Liberty, and Anke van {Zuylen}.
\newblock Improved approximation algorithms for bipartite correlation
  clustering.
\newblock {\em SIAM Journal on Computing}, 41(5):1110--1121, 2012.

\bibitem[AC11]{DBLP:journals/siamcomp/AilonC11}
Nir Ailon and Moses Charikar.
\newblock Fitting tree metrics: Hierarchical clustering and phylogeny.
\newblock {\em {SIAM} Journal on Computing}, 40(5):1275--1291, 2011.

\bibitem[ACN08]{ACN08}
Nir Ailon, Moses Charikar, and Alantha Newman.
\newblock Aggregating inconsistent information: ranking and clustering.
\newblock {\em Journal of the ACM (JACM)}, 55(5):1--27, 2008.

\bibitem[ADFH20]{aprile2020simple}
Manuel Aprile, Matthew Drescher, Samuel Fiorini, and Tony Huynh.
\newblock A simple 7/3-approximation algorithm for feedback vertex set in
  tournaments.
\newblock {\em CoRR}, arXiv abs/2008.08779, 2020.

\bibitem[AHK{\etalchar{+}}09]{agrawal2009generating}
Rakesh Agrawal, Alan Halverson, Krishnaram Kenthapadi, Nina Mishra, and
  Panayiotis Tsaparas.
\newblock Generating labels from clicks.
\newblock In {\em Proceedings of the Second ACM International Conference on Web
  Search and Data Mining}, pages 172--181, 2009.

\bibitem[AMMN06]{alon2006quadratic}
Noga Alon, Konstantin Makarychev, Yury Makarychev, and Assaf Naor.
\newblock Quadratic forms on graphs.
\newblock {\em Inventiones mathematicae}, 163(3):499--522, 2006.

\bibitem[ARS09]{arasu2009large}
Arvind Arasu, Christopher R{\'e}, and Dan Suciu.
\newblock Large-scale deduplication with constraints using dedupalog.
\newblock In {\em Proceedings of the 25th IEEE International Conference on Data
  Engineering (ICDE)}, pages 952--963, 2009.

\bibitem[AW22]{DBLP:conf/innovations/Assadi022}
Sepehr Assadi and Chen Wang.
\newblock Sublinear time and space algorithms for correlation clustering via
  sparse-dense decompositions.
\newblock In {\em Proceedings of the 13th Conference on Innovations in
  Theoretical Computer Science Conference (ITCS)}, volume 215 of {\em LIPIcs},
  pages 10:1--10:20, 2022.

\bibitem[BBC04]{BBC04}
Nikhil Bansal, Avrim Blum, and Shuchi Chawla.
\newblock Correlation clustering.
\newblock {\em Machine learning}, 56(1):89--113, 2004.

\bibitem[BCMT22]{abs-2205-03710}
Soheil Behnezhad, Moses Charikar, Weiyun Ma, and Li{-}Yang Tan.
\newblock Almost 3-approximate correlation clustering in constant rounds.
\newblock In {\em Proceedings of the 63rd Annual IEEE Symposium on Foundations
  of Computer Science {(FOCS)}}, pages 720--731, 2022.

\bibitem[BEK21]{bun2021differentially}
Mark Bun, Marek Elias, and Janardhan Kulkarni.
\newblock Differentially private correlation clustering.
\newblock In {\em International Conference on Machine Learning (ICML)}, pages
  1136--1146, 2021.

\bibitem[BGU13]{bonchi2013overlapping}
Francesco Bonchi, Aristides Gionis, and Antti Ukkonen.
\newblock Overlapping correlation clustering.
\newblock {\em Knowledge and Information Systems}, 35(1):1--32, 2013.

\bibitem[BRS11]{BRS11}
Boaz Barak, Prasad Raghavendra, and David Steurer.
\newblock Rounding semidefinite programming hierarchies via global correlation.
\newblock In {\em Proceedings of 52nd Annual IEEE Symposium on Foundations of
  Computer Science (FOCS)}, pages 472--481, 2011.

\bibitem[CDK14]{chierichetti2014correlation}
Flavio Chierichetti, Nilesh Dalvi, and Ravi Kumar.
\newblock Correlation clustering in mapreduce.
\newblock In {\em Proceedings of the 20th ACM SIGKDD international conference
  on Knowledge discovery and data mining}, pages 641--650, 2014.

\bibitem[CDK{\etalchar{+}}21]{DBLP:conf/focs/Cohen-Addad0KPT21}
Vincent Cohen{-}Addad, Debarati Das, Evangelos Kipouridis, Nikos Parotsidis,
  and Mikkel Thorup.
\newblock Fitting distances by tree metrics minimizing the total error within a
  constant factor.
\newblock In {\em Proceedings of 62nd {IEEE} Annual Symposium on Foundations of
  Computer Science (FOCS)}, pages 468--479, 2021.

\bibitem[CFL{\etalchar{+}}22]{DBLP:journals/corr/abs-2203-01440}
Vincent Cohen{-}Addad, Chenglin Fan, Silvio Lattanzi, Slobodan Mitrovic, Ashkan
  Norouzi{-}Fard, Nikos Parotsidis, and Jakub Tarnawski.
\newblock Near-optimal correlation clustering with privacy.
\newblock {\em CoRR}, arXiv abs/2203.01440, 2022.

\bibitem[CGS17]{charikar2017local}
Moses Charikar, Neha Gupta, and Roy Schwartz.
\newblock Local guarantees in graph cuts and clustering.
\newblock In {\em International Conference on Integer Programming and
  Combinatorial Optimization (IPCO)}, pages 136--147, 2017.

\bibitem[CGW05]{CGW05}
Moses Charikar, Venkatesan Guruswami, and Anthony Wirth.
\newblock Clustering with qualitative information.
\newblock {\em Journal of Computer and System Sciences}, 71(3):360--383, 2005.

\bibitem[CKK{\etalchar{+}}06]{CKKRS06}
Shuchi Chawla, Robert Krauthgamer, Ravi Kumar, Yuval Rabani, and D.~Sivakumar.
\newblock On the hardness of approximating multicut and sparsest-cut.
\newblock {\em Computational Complexity}, 15(2):94--114, 2006.

\bibitem[CKP08]{chakrabarti2008graph}
Deepayan Chakrabarti, Ravi Kumar, and Kunal Punera.
\newblock A graph-theoretic approach to webpage segmentation.
\newblock In {\em Proceedings of the 17th International conference on World
  Wide Web}, pages 377--386, 2008.

\bibitem[CLM{\etalchar{+}}21]{cohen2021correlation}
Vincent Cohen{-}Addad, Silvio Lattanzi, Slobodan Mitrovic, Ashkan
  Norouzi{-}Fard, Nikos Parotsidis, and Jakub Tarnawski.
\newblock Correlation clustering in constant many parallel rounds.
\newblock In {\em Proceedings of the 38th International Conference on Machine
  Learning (ICML)}, pages 2069--2078, 2021.

\bibitem[CLMP22]{Cohen-AddadLMP22}
Vincent Cohen{-}Addad, Silvio Lattanzi, Andreas Maggiori, and Nikos Parotsidis.
\newblock Online and consistent correlation clustering.
\newblock In {\em Proceedings of International Conference on Machine Learning
  (ICML)}, pages 4157--4179, 2022.

\bibitem[CMM09]{charikar2009integrality}
Moses Charikar, Konstantin Makarychev, and Yury Makarychev.
\newblock Integrality gaps for {S}herali-{A}dams relaxations.
\newblock In {\em Proceedings of 41st Annual ACM Symposium on Theory of
  Computing (STOC)}, pages 283--292, 2009.

\bibitem[CMSY15]{CMSY15}
Shuchi Chawla, Konstantin Makarychev, Tselil Schramm, and Grigory Yaroslavtsev.
\newblock Near optimal {LP} rounding algorithm for correlationclustering on
  complete and complete $k$-partite graphs.
\newblock In {\em Proceedings of the 47th Annual ACM Symposium on Theory of
  Computing (STOC)}, pages 219--228, 2015.

\bibitem[CSX12]{chen2012clustering}
Yudong Chen, Sujay Sanghavi, and Huan Xu.
\newblock Clustering sparse graphs.
\newblock In {\em Proceedings of the 25th International Conference on Neural
  Information Processing Systems-Volume 2}, pages 2204--2212, 2012.

\bibitem[CW04]{charikar2004maximizing}
Moses Charikar and Anthony Wirth.
\newblock Maximizing quadratic programs: Extending grothendieck's inequality.
\newblock In {\em Proceedings of 45th Annual IEEE Symposium on Foundations of
  Computer Science (FOCS)}, pages 54--60, 2004.

\bibitem[DEFI06]{DEFI06}
Erik~D. Demaine, Dotan Emanuel, Amos Fiat, and Nicole Immorlica.
\newblock Correlation clustering in general weighted graphs.
\newblock {\em Theoretical Computer Science}, 361(2-3):172--187, 2006.

\bibitem[GG06]{giotis2006correlation}
Ioannis Giotis and Venkatesan Guruswami.
\newblock Correlation clustering with a fixed number of clusters.
\newblock {\em Theory Of Computing}, 2:249--266, 2006.

\bibitem[GMT07]{gionis2007clustering}
Aristides Gionis, Heikki Mannila, and Panayiotis Tsaparas.
\newblock Clustering aggregation.
\newblock {\em ACM Transactions on Knowledge Discovery from Data}, 1(1):4,
  2007.

\bibitem[GMT09]{georgiou2009optimal}
Konstantinos Georgiou, Avner Magen, and Madhur Tulsiani.
\newblock Optimal {S}herali-{A}dams gaps from pairwise independence.
\newblock {\em Approximation, Randomization, and Combinatorial Optimization},
  pages 125--139, 2009.

\bibitem[GS11]{GS11}
Venkatesan Guruswami and Ali~Kemal Sinop.
\newblock Lasserre hierarchy, higher eigenvalues, and approximation schemes for
  graph partitioning and quadratic integer programming with {PSD} objectives.
\newblock In {\em Proceedings of 52nd Annual IEEE Symposium on Foundations of
  Computer Science (FOCS)}, pages 482--491, 2011.

\bibitem[HST20]{hopkins2020subexponential}
Samuel~B. Hopkins, Tselil Schramm, and Luca Trevisan.
\newblock Subexponential {LP}s approximate max-cut.
\newblock In {\em Proceedings of 61st Annual IEEE Symposium on Foundations of
  Computer Science (FOCS)}, pages 943--953, 2020.

\bibitem[JKMM20]{jafarov2020correlation}
Jafar Jafarov, Sanchit Kalhan, Konstantin Makarychev, and Yury Makarychev.
\newblock Correlation clustering with asymmetric classification errors.
\newblock In {\em International Conference on Machine Learning (ICML)}, pages
  4641--4650, 2020.

\bibitem[JKMM21]{jafarov2021local}
Jafar Jafarov, Sanchit Kalhan, Konstantin Makarychev, and Yury Makarychev.
\newblock Local correlation clustering with asymmetric classification errors.
\newblock In {\em International Conference on Machine Learning (ICML)}, pages
  4677--4686, 2021.

\bibitem[KCMNT08]{kalashnikov2008web}
Dmitri~V. Kalashnikov, Zhaoqi Chen, Sharad Mehrotra, and Rabia Nuray-Turan.
\newblock Web people search via connection analysis.
\newblock {\em IEEE Transactions on Knowledge and Data Engineering},
  20(11):1550--1565, 2008.

\bibitem[KMN11]{karlin2011integrality}
Anna~R. Karlin, Claire Mathieu, and C.~Thach Nguyen.
\newblock Integrality gaps of linear and semi-definite programming relaxations
  for knapsack.
\newblock In {\em International Conference on Integer Programming and
  Combinatorial Optimization (IPCO)}, pages 301--314, 2011.

\bibitem[KMZ19]{kalhan2019correlation}
Sanchit Kalhan, Konstantin Makarychev, and Timothy Zhou.
\newblock Correlation clustering with local objectives.
\newblock {\em Advances in Neural Information Processing Systems}, 32, 2019.

\bibitem[KS09]{karpinski2009linear}
Marek Karpinski and Warren Schudy.
\newblock Linear time approximation schemes for the {G}ale-{B}erlekamp game and
  related minimization problems.
\newblock In {\em Proceedings of the 41st Annual ACM Symposium on Theory of
  Computing (STOC)}, pages 313--322, 2009.

\bibitem[Liu22]{Daogao2022}
Daogao Liu.
\newblock Better private algorithms for correlation clustering.
\newblock {\em CoRR}, arXiv abs/2202.10747, 2022.

\bibitem[MSS10]{mathieu2010online}
Claire Mathieu, Ocan Sankur, and Warren Schudy.
\newblock Online correlation clustering.
\newblock In {\em Proceedings of 27th International Symposium on Theoretical
  Aspects of Computer Science (STACS)}, pages 573--584, 2010.

\bibitem[OS19]{o2019sherali}
Ryan O’Donnell and Tselil Schramm.
\newblock Sherali-{A}dams strikes back.
\newblock In {\em 34th Computational Complexity Conference}, 2019.

\bibitem[PM16]{puleo2016correlation}
Gregory Puleo and Olgica Milenkovic.
\newblock Correlation clustering and biclustering with locally bounded errors.
\newblock In {\em International Conference on Machine Learning}, pages
  869--877. PMLR, 2016.

\bibitem[RT12]{RT12}
Prasad Raghavendra and Ning Tan.
\newblock Approximating {CSPs} with global cardinality constraints using {SDP}
  hierarchies.
\newblock In {\em Proceedings of the 23rd Annual ACM-SIAM Symposium on Discrete
  Algorithms (SODA)}, pages 373--387, 2012.

\bibitem[Swa04]{swamy2004correlation}
Chaitanya Swamy.
\newblock Correlation clustering: {M}aximizing agreements via semidefinite
  programming.
\newblock In {\em Proceedings of the 15th Annual ACM-SIAM Symposium on Discrete
  Algorithms (SODA)}, pages 526--527, 2004.

\bibitem[VZW09]{van2009deterministic}
Anke Van~Zuylen and David~P. Williamson.
\newblock Deterministic pivoting algorithms for constrained ranking and
  clustering problems.
\newblock {\em Mathematics of Operations Research}, 34(3):594--620, 2009.

\bibitem[YZ14]{yoshida2014approximation}
Yuichi Yoshida and Yuan Zhou.
\newblock Approximation schemes via {S}herali-{A}dams hierarchy for dense
  constraint satisfaction problems and assignment problems.
\newblock In {\em Proceedings of the 5th Conference on Innovations in
  Theoretical Computer Science (ITCS)}, pages 423--438, 2014.

\end{thebibliography}

\end{document}